\documentclass[runningheads]{llncs}

\usepackage[ruled,longend,linesnumbered,noend]{algorithm2e}
\usepackage[T1]{fontenc}

\usepackage{graphicx} % Required for inserting images
\usepackage{tcolorbox}
\usepackage{todonotes}
\usepackage{amssymb}
\usepackage{enumerate}
\usepackage{amsmath}
\usepackage{orcidlink}

\usepackage{xcolor}
\definecolor{mynicegreen}{RGB}{70,150,70}

\newcommand{\rtd}{{\sc Roman \{3\}-Domination}}
\newcommand{\drd}{{\sc Double Roman Domination}}
\newcommand{\rrd}{{\sc Roman \{2\}-Domination}}
\newcommand{\trd}{{\sc Total Roman Domination}}
\newcommand{\ird}{{\sc Independent Roman Domination}}
\newcommand{\wrd}{{\sc Weak Roman Domination}}
\newcommand{\rd}{{\sc Roman Domination}}
\newcommand{\ds}{{\sc Dominating Set}}

\newcommand{\ethc}{{\sc Exact 3-cover}}

\newtheorem{observation}{Observation}{\bfseries}{\itshape}

\begin{document}
\title{Hardness and Algorithmic Results for \newline Roman \{3\}-Domination}
\titlerunning{Roman \{3\}-Domination}

\author{Sangam Balchandar Reddy\orcidlink{0000-0002-5848-3821}}
% %
\authorrunning{S. B. Reddy}
\institute{School of Computer and Information Sciences, \\ University of Hyderabad, Hyderabad, India \vspace{2mm} \\
\email{21mcpc14@uohyd.ac.in}}
\maketitle              % typeset the header of the contribution
\begin{abstract}
A Roman $\{3\}$-dominating function on a graph $G = (V, E)$ is a function $f: V \rightarrow \{0, 1, 2, 3\}$ such that for each vertex $u \in V$, if $f(u) = 0$ then $\sum_{v \in N(u)} f(v) \geq 3$ and if $f(u) = 1$ then $\sum_{v \in N(u)} f(v) \geq 2$. The weight of a Roman $\{3\}$-dominating function $f$ is $\sum_{u \in V} f(u)$. The objective of \rtd{} is to compute a Roman $\{3\}$-dominating function of minimum weight. The problem is known to be NP-complete on chordal graphs, star-convex bipartite graphs and comb-convex bipartite graphs. In this paper, we study the complexity of \rtd{} and show that the problem is NP-complete on split graphs. In addition, we prove that the problem is W[2]-hard parameterized by weight. On the positive front, we present a polynomial-time algorithm for block graphs, thereby resolving an open question posed by Chaudhary and Pradhan [Discrete Applied Mathematics, 2024]. 
\end{abstract}
\section{Introduction}
Given a graph $G = (V, E)$, a Roman dominating function $f: V \rightarrow{} \{0 ,1, 2\}$ is a labeling of vertices such that each vertex $u \in V$ with $f(u) = 0$ has a vertex $v \in N(u)$ with $f(v) =2$. The weight of a Roman dominating function $f$ is the sum of $f(u)$ over all $u \in V(G)$, that is, $\sum_{u \in V(G)} f(u)$. The objective of \rd{} (RD) is to compute a Roman dominating function of minimum weight.

RD is a fascinating concept in graph theory that draws inspiration from the strategic positioning of legions to protect the Roman Empire. The concept of RD has been introduced by Cockayene et al.~\cite{COCKAYNE200411} in 2004. Liu et al.~\cite{liu2013roman} proved that RD is NP-complete on split graphs and bipartite graphs. They also provided a linear-time algorithm for strongly chordal graphs. RD also admits linear-time algorithms for cographs, interval graphs, graphs of bounded clique-width and a polynomial-time algorithm for AT-free graphs~\cite{liedloff2008efficient}. 

Fernau et al.~\cite{henning} studied RD in the realm of parameterized complexity and proved that the problem is W[2]-complete parameterized by weight. They also provided an FPT algorithm that runs in time $\mathcal{O}^*(5^t)$, where $t$ is the treewidth of the input graph. In addition, on planar graphs, they showed that RD can be solved in $\mathcal{O}^*(3.3723^k)$ time. Ashok et al.~\cite{ashok2024independent} provided an $\mathcal{O}^*(4^d)$ algorithm for RD where $d$ is the distance to cluster. Recently, Mohanapriya et al.~\cite{mohanapriya2023roman} proved that RD parameterized by solution size is W[1]-hard, even on split graphs.

Many variants of RD, such as \rrd{}, \drd{}, \trd{}, \ird{} and \wrd{} were studied from an algorithmic standpoint. Now, we discuss the problems \rrd{} and \drd{}, as well as their computational complexities.

Chellali et al.~\cite{CHELLALI201622} introduced the concept of \rrd{} in graphs. For a graph $G = (V, E)$, a Roman $\{2\}$-dominating function $f: V \rightarrow{} \{0 ,1, 2\}$ is a labeling of vertices such that for each vertex $u \in V$ with $f(u) = 0$, either there exist two vertices $v_1, v_2 \in N(u)$ with $f(v_1) =1$ and $f(v_2) = 1$ or there exists a vertex $v \in N(u)$ with $f(v) =2$. The results related to the complexity of the problem can be found in~\cite{chen2022roman,fernandez2023new,poureidi2020algorithmic}. 

The study on \drd{} was initiated by Beeler et al.~\cite{beeler2016double}. A double Roman dominating function $f: V \rightarrow{} \{0 ,1, 2, 3\}$ is a labeling of vertices such that for each vertex $u \in V$ with $f(u) = 0$, either there exist two vertices $v_1, v_2 \in N(u)$ with $f(v_1) =2$ and $f(v_2) = 2$ or there exists a vertex $v \in N(u)$ with $f(v) =3$. For each vertex $u \in V$ with $f(u) = 1$ there exists a vertex $v \in N(u)$ with $f(v) \geq2$. More details on the problem complexity can be found in~\cite{ahangar2017double,banerjee2020algorithmic,yue2018double,zhang2018double}.

A Roman $\{3\}$-dominating function on a graph $G = (V, E)$ is a function $f: V \rightarrow \{0, 1, 2, 3\}$ such that for each vertex $u \in V$, if $f(u) = 0$ then $\sum_{v \in N(u)} f(v) \geq 3$, else if $f(u) = 1$ then $\sum_{v \in N(u)} f(v) \geq 2$. The weight of a Roman \{3\}-dominating function $f$ is $\sum_{u \in V(G)} f(u)$. The minimum weight over all Roman \{3\}-dominating functions is denoted by $\gamma_{R3}(G)$. The objective of \rtd{} (R3D) is to compute $\gamma_{R3}(G)$.

% The decision version of \rtd{} is defined as follows.
% \begin{tcolorbox}
% {
% \rtd{}: \vspace{2mm} \newline
% \textit{Input:} An instance $I$ = $(G, k)$, where $G=(V,E)$ is an undirected graph, and an integer $k$.\newline
% \textit{Output:} YES, if there exists a function $f: V \rightarrow{} \{0,1,2,3\}$ such that \\
% (1) $w(f) \leq k$, \\
% (2) If $f(u) = 0$, then $\sum_{v \in N(u)} f(v) \geq 3$, else if $f(u) =1$, then $\sum_{v \in N(u)} f(v) \geq 2$; NO otherwise.
% }
% \end{tcolorbox}
Mojdeh et al.~\cite{mojdeh2020roman} initiated the study on R3D by showing that the problem is NP-complete on bipartite graphs. Chakradar et al.~\cite{chakradhar2022algorithmic} proved that the problem is NP-complete on chordal graphs, planar graphs, star-convex bipartite graphs and comb-convex bipartite graphs. They provided linear-time algorithms for bounded treewidth graphs, chain graphs and threshold graphs. They also showed that the problem is APX-complete on graphs with maximum degree 4. Goyal et al.~\cite{goyal2021hardness} proposed an $\mathcal{O}(\ln \Delta(G))$ approximation algorithm. Recently, Chaudhary et al.~\cite{CHAUDHARY2024301} proved that the problem is NP-complete on chordal bipartite graphs, planar graphs and star-convex bipartite graphs. They also presented linear-time algorithms for chain graphs and cographs. In addition, they provided several approximation-related results.

In this paper, we extend the study on the complexity aspects of R3D and obtain the following results. 
\begin{enumerate}
    \item As the problem is known to be NP-complete on chordal graphs, we study the problem complexity on \textit{split graphs}, a subclass of chordal graphs and prove that the problem is NP-complete.
    \item We initiate the study on parameterized complexity of R3D and prove that the problem is W[2]-hard parameterized by weight.
    \item In addition, we present a polynomial-time algorithm for block graphs.
\end{enumerate}
\section{Preliminaries}
Let $G = (V, E)$ be a graph with $V$ as the vertex set and $E$ as the edge set, such that $n = |V|$ and $m = |E|$. The set of vertices that belong to $N(u)$ and $N[u]$, respectively, are referred to as the neighbours and closed neighbours of $u$. The open neighbourhood of a set $T \subseteq V$ is denoted by $N(T)$ and the closed neighbourhood by $N[T]$. $N(T)$ = $\bigcup\limits_{u \in T}^{} N(u)$ and $N[T]$ = $\bigcup\limits_{u \in T}^{} N[u]$. For any vertex $u \in V$, we use the term \textbf{label} to denote the value assigned to $u$ by the function $f$, i.e., $f(u)$. Similarly, for any vertex $u \in V$, we define the term \textbf{labelSum} as the sum of the labels of all the vertices in the closed neighbourhood of $u$, i.e., $\sum_{v \in N[u]} f(v)$. For a vertex subset $T \subseteq V$, the \textbf{weight} of $T$, denoted by $f(T)$, is the sum of the labels of all the vertices in $T$, i.e., $\sum_{u \in T} f(u)$. We say that a vertex $u \in V$ with $f(u)$ $\in$ $\{0,1\}$ is \textit{dominated} if it satisfies the Roman \{3\}-dominating function constraint.

A problem is considered to be \textit{fixed-parameter tractable} w.r.t. a parameter $k$, if there exists an algorithm with running time $\mathcal{O}(f(k)n^{\mathcal{O}(1)})$, where $f$ is some computable function. A problem is said to be W[1]-hard with respect to a parameter if it is believed not to be fixed-parameter tractable (FPT), serving as an analogue to NP-hardness in classical complexity theory. 
% In parameterized complexity, the hierarchy of complexity classes is $\text{FPT}$ $\subseteq \text{W[1]}$ $\subseteq \text{W[2]}$ $\subseteq$ $\cdots$ $\subseteq \text{XP}$. 
For more information on \textit{graph thery} and \textit{parameterized complexity}, we refer the reader to~\cite{west} and~\cite{MC}, respectively.
%  given as follows.
% \[
% \text{FPT} \subseteq \text{W[1]} \subseteq \text{W[2]} \subseteq \cdots \subseteq \text{XP}.
% \]

Now, we define the graph classes that are used in this paper. The graph class \textit{split graphs} is defined as follows. 
\begin{definition} [Split graphs]
A graph $G = (V, E)$ is a \textit{split graph} if the vertex set $V$ can be partitioned into two sets $V_1$ and $V_2$, such that the induced graph $G[V_1]$ is a complete graph and the induced graph $G[V_2]$ is an independent set.
\end{definition}
The graph class \textit{block graphs} is defined as follows. 
\begin{definition} [Block graphs]
A vertex $u$ is a cut-vertex of $G$ if the number of components in $G-\{u\}$ is higher than that of $G$. A block of $G$ is a maximal connected subgraph of $G$ that does not have a cut-vertex. A \textit{block graph} is a simple graph in which every block is a complete graph.
\end{definition}
 
Let \( \{c_1, c_2, \ldots, c_\ell\} \) be the set of cut-vertices and \( \{\mathcal{B}_1, \mathcal{B}_2, \ldots, \mathcal{B}_r\} \) be the set of blocks of a block graph \( G \). The \emph{cut-tree} of \( G \), denoted by \( T_G \), is defined as the graph with vertex set
$V(T_G)$ = $\{\mathcal{B}_1, \mathcal{B}_2, \ldots, \mathcal{B}_r, c_1, c_2, \ldots, c_\ell\}$
and an edge set
$E(T_G)$ = $\{\mathcal{B}_i c_j : c_j \in \mathcal{B}_i,\ i \in [r],\ j \in [\ell]\}.$ A block graph $G$ and its corresponding \textit{cut-tree} $T_G$ are given in Fig. \ref{fig:fig1}. An \textit{end block} of a block graph is a block that contains exactly one cut-vertex.
\begin{figure} 
    \centering \vspace{-4mm}
    \begin{tikzpicture} [scale =0.52]
        \filldraw (0, 2) circle (3pt) node[anchor=south]{};
        \filldraw (0, 4) circle (3pt) node[anchor=south]{};
        \filldraw (2, 3) circle (3pt) node[anchor=south]{};
        \filldraw (2, 5) circle (3pt) node[anchor=south]{};
        
        \filldraw (4, 2) circle (3pt) node[anchor=south]{};
        \filldraw (4, 4) circle (3pt) node[anchor=south]{};
        \filldraw (6, 3) circle (3pt) node[anchor=south]{};
        \filldraw (8, 3) circle (3pt) node[anchor=south]{};

        \filldraw (0, 2) circle (3pt) node[anchor=south]{};
        \filldraw (0, 4) circle (3pt) node[anchor=south]{};
        \filldraw (2, 3) circle (3pt) node[anchor=south]{};
        \filldraw (2, 5) circle (3pt) node[anchor=south]{};
        
        \filldraw (4, 2) circle (3pt) node[anchor=south]{};
        \filldraw (4, 4) circle (3pt) node[anchor=south]{};
        \filldraw (6, 3) circle (3pt) node[anchor=south]{};
        \filldraw (8, 3) circle (3pt) node[anchor=south]{};

        \draw [thin] (0, 2) -- (0, 4);\draw [thin] (0, 2) -- (2, 3);
        \draw [thin] (0, 4) -- (2, 3);\draw [thin] (2, 3) -- (2, 5);
        \draw [thin] (2, 3) -- (4, 4);\draw [thin] (2, 3) -- (4, 2);
        \draw [thin] (2, 3) -- (6, 3);\draw [thin] (4, 4) -- (4, 2);
        \draw [thin] (4, 4) -- (6, 3);\draw [thin] (4, 2) -- (6, 3);
        \draw [thin] (6, 3) -- (8, 3);

        \filldraw (13, 4) circle (3pt) node[anchor=south]{};
        \filldraw (15, 4) circle (3pt) node[anchor=south]{};
        
        \filldraw (12, 2) circle (3pt) node[anchor=south]{};
        \filldraw (14, 2) circle (3pt) node[anchor=south]{};
        \filldraw (16, 2) circle (3pt) node[anchor=south]{};
        \filldraw (14, 6) circle (3pt) node[anchor=south]{};

        \draw [thin] (13, 4) -- (14, 6);\draw [thin] (14, 6) -- (15, 4);
        \draw [thin] (15, 4) -- (16, 2);\draw [thin] (13, 4) -- (12, 2);
        \draw [thin] (13, 4) -- (14, 2);

        \filldraw (-0.5, 1.75) circle (0pt) node[anchor=south]{$v_1$};
        \filldraw (-0.5, 3.75) circle (0pt) node[anchor=south]{$v_2$};
        \filldraw (2, 2.05) circle (0pt) node[anchor=south]{$v_3$};
        \filldraw (4, 4.15) circle (0pt) node[anchor=south]{$v_6$};
        \filldraw (6, 3.15) circle (0pt) node[anchor=south]{$v_8$};
        \filldraw (8, 3.15) circle (0pt) node[anchor=south]{$v_9$};
        \filldraw (4, 1.15) circle (0pt) node[anchor=south]{$v_7$};
        \filldraw (2, 5.15) circle (0pt) node[anchor=south]{$v_4$};
        
        \filldraw (11.5, 1.8) circle (0pt) node[anchor=south]{$\mathcal{B}_1$};
        \filldraw (14.5, 1.8) circle (0pt) node[anchor=south]{$\mathcal{B}_2$};
        \filldraw (12.6, 3.9) circle (0pt) node[anchor=south]{$v_3$};
        \filldraw (15.4, 3.9) circle (0pt) node[anchor=south]{$v_8$};
        \filldraw (14, 6.1) circle (0pt) node[anchor=south]{$\mathcal{B}_3$};
        \filldraw (16.5, 1.8) circle (0pt) node[anchor=south]{$\mathcal{B}_4$};
        
        \filldraw (3.5, 0.15) circle (0pt) node[anchor=south]{$G$};
        \filldraw (14, 0.15) circle (0pt) node[anchor=south]{$T_G$};
    \end{tikzpicture} \vspace{-1mm}
    \caption{Block graph $G$ and the corresponding \textit{cut-tree} $T_G$}
    \label{fig:fig1}
\end{figure}
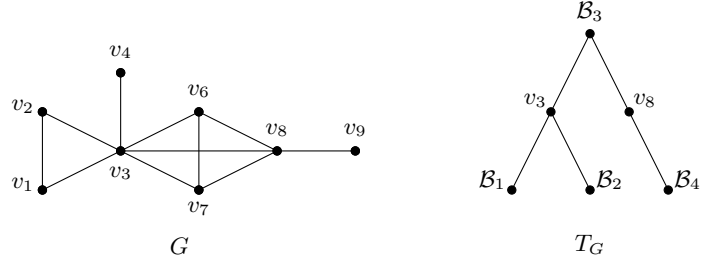
\section{NP-complete on split graphs}
R3D was proved to be NP-complete on chordal graphs~\cite{chakradhar2022algorithmic,CHAUDHARY2024301}. In this section, we strengthen the known result on chordal graphs by proving that R3D is NP-complete on split graphs, a subclass of chordal graphs. We provide a polynomial-time reduction from the \ethc{} (X3C) problem to prove that R3D is NP-complete on split graphs. The X3C problem is defined as follows.
% \begin{tcolorbox}
% {

\medskip
\noindent \ethc{}: \vspace{2mm} \newline
\textit{Input:} A set $\overline{X}$ with $|\overline{X}| = 3q$ and a collection $\overline{C}$ of 3-element subsets of $\overline{X}$.\vspace{1mm} \\ 
\textit{Output:} YES, if there exists a collection of sets $\overline{C'}$ of $\overline{C}$ such that every element in $\overline{X}$ occurs in
exactly one member of $\overline{C'}$; NO otherwise.
% }
% \end{tcolorbox}
% \vspace{2mm}

\medskip
\noindent The complexity result related to X3C is given as follows.
\begin{theorem}[\cite{johnson1979computers}]~\label{etcresult}
    X3C is NP-complete.
\end{theorem}
% \noindent X3C is NP-complete~\cite{johnson1979computers}. 
% \vspace{2mm} \\
Consider an instance $\langle \overline{X}, \overline{C}\rangle$ of X3C, where $\overline{X} = \{\overline{x}_1, \overline{x}_2, ...,\overline{x}_{3q}\}$ and $\overline{C} = \{\overline{C}_1, \overline{C}_2, ...,\overline{C}_t\}$. For the sake of generality, we assume that $q$ is even. This is valid since any odd value of $q$ can be made even by appending three dummy elements $\overline{x}_{3q+1}, \overline{x}_{3q+2}$ and $\overline{x}_{3q+3}$ to the set $\overline{X}$ and by augmenting $\overline{C}$ with an additional 3-element subset $\overline{C}_{t+1}$ $=$ $\{\overline{x}_{3q+1}, \overline{x}_{3q+2}, \overline{x}_{3q+3}\}$. From an instance $\langle \overline{X}, \overline{C}\rangle$ of X3C, we construct an instance $I = (G, k)$ of R3D as follows.
\begin{itemize}
    \item For each $\overline{x}_i\in \overline{X}$, we introduce a vertex $x_i$ in $G$. Similarly, for each $\overline{C}_j \in \overline{C}$, we introduce a vertex $c_j$ in $G$. Let $X = \{x_1, x_2, ..., x_{3q}\}$ and $C = \{c_1, c_2, ..., c_t\}$.
    \item We create two additional copies of the vertex set $X$, denoted by, $A$ and $B$. A copy of vertex $x_i \in X$ is denoted by $a_i$ in $A$ and $b_i $ in $B$.
    \item We create two sets $Y$ and $Z$ of size $10q$ each. Let $Y = \{y_1,y_2,...,y_{10q}\}$ and $Z = \{z_1,z_2,...,z_{10q}\}$.
    \item For each $x_i \in X$ and $c_j \in C$, we add an edge between $x_i$ and $c_j$ if and only if $\overline{x}_i \in \overline{C}_j$.
    \item We make vertex $x_i$ adjacent to both $a_i$ and $b_i$, for each $i \in [3q]$.
    \item We construct the adjacency between the vertices of the sets $A$ and $Z$ as follows. We arbitrarily select a subset $A_i \subseteq A$ of six vertices, and a subset $Z_j \subseteq Z$ of twenty vertices. We make each vertex in $Z_j$ adjacent to a unique triplet of vertices in $A_i$. This process is repeated for each such pair of subsets $A_i$ and $Z_j$ chosen in successive iterations. Importantly, in each iteration, the sets $A_i$ and $Z_j$ should not include any vertices from $A$ and $Z$, respectively, that were selected in previous iterations. In other words, vertices in the sets $A \cup Z$ must participate exactly once in this adjacency construction. In a similar manner, we establish the adjacency between the sets $B$ and $Y$.
    \item Additionally, we make the set $A \cup B \cup C$ a clique.
\end{itemize}
\begin{figure} [t]
    \centering
     \begin{tikzpicture} [thick,scale=1.5, every node/.style={scale=0.95}]
     \draw (0, 0) ellipse (1 and 2);
        \draw (4,0) ellipse (1 and 2);
        \draw[thin, gray] (-0.97, 0.6) -- (0.97, 0.6);
        \draw[thin, gray] (-0.95, -0.7) -- (0.95, -0.7);
        \draw[thin, gray] (3.05, 0.6) -- (4.95, 0.6);
        \draw[thin, gray] (3.05, -0.6) -- (4.95, -0.6);

        \draw[blue] (4, 1.75) -- (0, 0.4);
        \draw[thin, gray] (4, 1.575) -- (0, 0.225);
        \draw[thin, gray] (4, 1.4) -- (0, -0.125);
        \draw[thin, gray] (4, 1.225) -- (0, -0.3);
        \draw[thin, gray] (4, 1.05) -- (0, 0.05);
        \draw[thin, gray] (4, 0.875) -- (0, -0.475);

        \draw[red] (4, 0.375) -- (0, 0.4);
        \draw[red] (4, 0.375) -- (0, 0.225);
        \draw[red] (4, 0.375) -- (0, 0.05);        
        \draw[thin, gray] (4, 0.2) -- (0, 0.4);
        \draw[thin, gray] (4, 0.2) -- (0, 0.225);
        \draw[thin, gray] (4, 0.2) -- (0, -0.125);
        \draw[thin, gray] (4, -0.4) -- (0, -0.3);
        \draw[thin, gray] (4, -0.4) -- (0, -0.125);
        \draw[thin, gray] (4, -0.4) -- (0, -0.475);
        
        \draw[mynicegreen] (4, -0.875) -- (0, -0.875);
        \draw[mynicegreen] (4, -0.875) -- (0, -1.05);
        \draw[mynicegreen] (4, -0.875) -- (0, -1.225);
        \draw[thin, gray] (4, -1.0625) -- (0, -0.875);
        \draw[thin, gray] (4, -1.0625) -- (0, -1.05);
        \draw[thin, gray] (4, -1.0625) -- (0, -1.4);
        \draw[thin, gray] (4, -1.7875) -- (0, -1.4);
        \draw[thin, gray] (4, -1.7875) -- (0, -1.575);
        \draw[thin, gray] (4, -1.7875) -- (0, -1.75);

        \draw[blue] (4, 1.75) -- (0, -0.875);
        \draw[thin, gray] (4, 1.575) -- (0, -1.05);
        \draw[thin, gray] (4, 1.4) -- (0, -1.225);
        \draw[thin, gray] (4, 1.225) -- (0, -1.4);
        \draw[thin, gray] (4, 1.05) -- (0, -1.575);
        \draw[thin, gray] (4, 0.875) -- (0, -1.75);

        \draw[blue] (4, 1.75) -- (0, 1.7);
        \draw[blue] (4, 1.75) -- (0, 1.4);
        \draw[thin, gray] (4, 1.575) -- (0, 1.7);
        \draw[thin, gray] (4, 1.575) -- (0, 1.4);
        \draw[thin, gray] (4, 1.575) -- (0, 0.8);
        \draw[thin, gray] (4, 1.4) -- (0, 1.7);
        \draw[thin, gray] (4, 1.225) -- (0, 1.4);
        \draw[thin, gray] (4, 1.225) -- (0, 1.1);
        \draw[thin, gray] (4, 1.05) -- (0, 1.1);
        \draw[thin, gray] (4, 1.05) -- (0, 0.8);
        \draw[thin, gray] (4, 0.875) -- (0, 1.1);
        \draw[thin, gray] (4, 0.875) -- (0, 0.8);   

        \filldraw (0, 1.7) circle (1pt) node[anchor=south]{};
        \filldraw (0, 1.1) circle (1pt) node[anchor=south]{};
        \filldraw (0, 1.4) circle (1pt) node[anchor=south]{};
        \filldraw (0, 0.8) circle (1pt) node[anchor=south]{};

        \filldraw (-0.25, 1.55) circle (0cm) node[anchor=south]{$c_1$};
        \filldraw (-0.25, 1.25) circle (0cm) node[anchor=south]{$c_2$};
        \filldraw (-0.25, 0.95) circle (0cm) node[anchor=south]{$c_3$};
        \filldraw (-0.25, 0.65) circle (0cm) node[anchor=south]{$c_4$};  

        \filldraw (4.25, 1.6) circle (0cm) node[anchor=south]{$x_1$};
        \filldraw (4.25, 1.425) circle (0cm) node[anchor=south]{$x_2$};
        \filldraw (4.25, 1.25) circle (0cm) node[anchor=south]{$x_3$};
        \filldraw (4.25, 1.075) circle (0cm) node[anchor=south]{$x_4$};
        \filldraw (4.25, 0.9) circle (0cm) node[anchor=south]{$x_5$};
        \filldraw (4.25, 0.725) circle (0cm) node[anchor=south]{$x_6$};

        \filldraw (-0.25, 0.275) circle (0cm) node[anchor=south]{$b_1$};
        \filldraw (-0.25, 0.1) circle (0cm) node[anchor=south]{$b_2$};
        \filldraw (-0.25, -0.075) circle (0cm) node[anchor=south]{$b_3$};
        \filldraw (-0.25, -0.25) circle (0cm) node[anchor=south]{$b_4$};
        \filldraw (-0.25, -0.425) circle (0cm) node[anchor=south]{$b_5$};
        \filldraw (-0.25, -0.6) circle (0cm) node[anchor=south]{$b_6$};     

        \filldraw (-0.25, -0.975) circle (0cm) node[anchor=south]{$a_1$};
        \filldraw (-0.25, -1.15) circle (0cm) node[anchor=south]{$a_2$};
        \filldraw (-0.25, -1.325) circle (0cm) node[anchor=south]{$a_3$};
        \filldraw (-0.25, -1.5) circle (0cm) node[anchor=south]{$a_4$};
        \filldraw (-0.25, -1.675) circle (0cm) node[anchor=south]{$a_5$};
        \filldraw (-0.25, -1.85) circle (0cm) node[anchor=south]{$a_6$}; 
        
        \filldraw (4.25, 0.25) circle (0cm) node[anchor=south]{$y_1$};
        \filldraw (4.25, 0.05) circle (0cm) node[anchor=south]{$y_2$};
        \filldraw (4.25, -0.525) circle (0cm) node[anchor=south]{$y_{20}$};
        \filldraw (4.25, -1.2) circle (0cm) node[anchor=south]{$z_2$};
        \filldraw (4.25, -1) circle (0cm) node[anchor=south]{$z_1$};
        \filldraw (4.25, -1.9) circle (0cm) node[anchor=south]{$z_{20}$};
        
        \filldraw (0, 0.4) circle (1pt) node[anchor=south]{};
        \filldraw (0, 0.225) circle (1pt) node[anchor=south]{};
        \filldraw (0, 0.05) circle (1pt) node[anchor=south]{};
        \filldraw (0, -0.125) circle (1pt) node[anchor=south]{};
        \filldraw (0, -0.3) circle (1pt) node[anchor=south]{};
        \filldraw (0, -0.475) circle (1pt) node[anchor=south]{};

        \filldraw (0, -0.875) circle (1pt) node[anchor=south]{};
        \filldraw (0, -1.05) circle (1pt) node[anchor=south]{};
        \filldraw (0, -1.225) circle (1pt) node[anchor=south]{};
        \filldraw (0, -1.4) circle (1pt) node[anchor=south]{};
        \filldraw (0, -1.575) circle (1pt) node[anchor=south]{};
        \filldraw (0, -1.75) circle (1pt) node[anchor=south]{};

        \filldraw (4, 1.75) circle (1pt) node[anchor=south]{};
        \filldraw (4, 1.575) circle (1pt) node[anchor=south]{};
        \filldraw (4, 1.4) circle (1pt) node[anchor=south]{};
        \filldraw (4, 1.225) circle (1pt) node[anchor=south]{};
        \filldraw (4, 1.05) circle (1pt) node[anchor=south]{};
        \filldraw (4, 0.875) circle (1pt) node[anchor=south]{};
        
        \filldraw (4, 0.375) circle (1pt) node[anchor=south]{};
        \filldraw (4, 0.2) circle (1pt) node[anchor=south]{};
        \filldraw (4, -0.4) circle (1pt) node[anchor=south]{};

        \filldraw (4, -0.025) circle (0.5pt) node[anchor=south]{};
        \filldraw (4, -0.1) circle (0.5pt) node[anchor=south]{};
        \filldraw (4, -0.175) circle (0.5pt) node[anchor=south]{};

        \filldraw (4, -1.05) circle (1pt) node[anchor=south]{};
        \filldraw (4, -0.875) circle (1pt) node[anchor=south]{};
        \filldraw (4, -1.8) circle (1pt) node[anchor=south]{};
        
        \filldraw (4, -1.375) circle (0.5pt) node[anchor=south]{};
        \filldraw (4, -1.45) circle (0.5pt) node[anchor=south]{};
        \filldraw (4, -1.525) circle (0.5pt) node[anchor=south]{};
        
        \filldraw (-1.2, 1) circle (0cm) node[anchor=south]{$C$};
        \filldraw (-1.3, -0.2) circle (0cm) node[anchor=south]{$B$};
        \filldraw (-1.2, -1.6) circle (0cm) node[anchor=south]{$A$};
        \filldraw (0.1, -2.5) circle (0cm) node[anchor=south]{clique};
        \filldraw (4.2, -2.5) circle (0cm) node[anchor=south]{independent set};
        \filldraw (5.2, 1) circle (0cm) node[anchor=south]{$X$};
        \filldraw (5.3, -0.2) circle (0cm) node[anchor=south]{$Y$};
        \filldraw (5.2, -1.6) circle (0cm) node[anchor=south]{$Z$};
        \filldraw (2, -3.1) circle (0cm) node[anchor=south]{$G$};
     \end{tikzpicture}
    \caption{Reduced instance $G$ for an instance of X3C where $\overline{X} = \{1,2,3,4,5,6\}$ and $\overline{C}$ = $\{\overline{C_1} =  \{1,2,3\}$, $\overline{C_2} =  \{1,2,4\}$, $\overline{C_3} =  \{4,5,6\}$ and $\overline{C_4} =  \{2,5,6\}\}$. Edges between the vertices of the clique $A \cup B \cup C$ are not shown in the figure. The edges incident on the vertices $x_1$, $y_1$ and $z_1$ are colored in blue, red and green respectively.}
    \label{fig:fig2} \vspace{-3mm}
\end{figure}
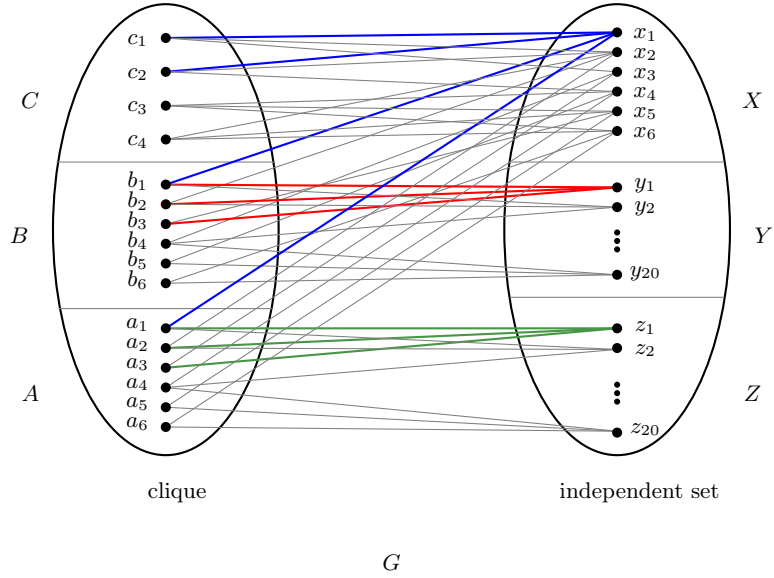
This concludes the construction of $G$. See Fig. \ref{fig:fig2} for an illustration. \vspace{2mm} \\
Observe that $G$ is a split graph, where the subgraph induced by $A \cup B \cup C$ forms a clique and the set $X \cup Y \cup Z$ is an independent set in $G$. 
\begin{lemma}~\label{lemma1}
    Let $f$ be a minimum weighted Roman \{3\}-dominating function for $G$. Then $f(A) \geq 3q$ and $f(B) \geq 3q$.
\end{lemma}
\begin{proof}
    Let $A_i = \{a_1, a_2, a_3, a_4, a_5, a_6\}$ be a six-vertex subset of $A$. By construction, set $A_i$ is associated with $Z_j = \{z_1, z_2,...,z_{20}\}$ in $Z$. Each vertex in $Z_j$ is adjacent to a distinct triple of vertices from $A_i$. Importantly, the vertices of $Z$ are adjacent only to the vertices of $A$. Therefore, optimally, no vertex from $Z$ gets a label from $\{1,2,3\}$. To ensure the required labelSum for the vertices in $Z_j$, every triple $(a_i, a_j, a_k)$ from $A_i$ must satisfy $f(a_i)+f(a_j)+f(a_k) \geq 3$. The argument holds for every six-vertex subset of $A$ and $B$. Therefore, we conclude that $f(A) \geq 3q$ and similarly, we have $f(B) \geq 3q$. \qed
\end{proof}
\begin{lemma}~\label{lemma2}
    $\langle \overline{X}, \overline{C}\rangle $ is a yes-instance of X3C if and only if $I$ has a Roman \{3\}-dominating function of weight $7q$.
\end{lemma}
\begin{proof}
    $[\Rightarrow]$ Let $\overline{C'} \subseteq \overline{C}$ be a solution to a yes-instance $\langle \overline{X}, \overline{C}\rangle $ of X3C. We construct a Roman \{3\}-dominating function $f$ as follows.
        \[
f(v) = 
\begin{cases}
0, & \text{if } v \in X \cup Y \cup Z  \cup \bigcup_{\overline{C}_j \notin \overline{C'}} c_j, \\
1, & \text{if } v \in A \cup B \cup \bigcup_{\overline{C}_j \in \overline{C'}} c_j, \\
2, & \text{if } v \in \emptyset, \\
3, & \text{if } v \in \emptyset
\end{cases}
\]
Each vertex in $X \cup Y \cup Z$ is dominated because it has three neighbours in $A \cup B \cup C$ with a label of 1. Since $A \cup B \cup C$ forms a clique, every vertex in it is dominated. The weight of $A \cup B$ is $6q$ and the weight of $C$ is $q$. Hence, we conclude that $f$ is a Roman \{3\}-dominating function of weight $7q$. \vspace{2mm} \\
$[\Leftarrow]$ From Lemma \ref{lemma1}, we have that $f(A \cup B) \geq $ $6q$. At this point, each vertex of the set $A \cup B \cup C \cup Y \cup Z$ is dominated. The vertices of $X$ are yet to be dominated. Given that the overall weight of the Roman \(\{3\}\)-dominating function is \( 7q \), this leaves at most a weight of \( q \) for the remaining vertices in \( X \cup C \). To ensure that all \( 3q \) vertices in \( X \) are dominated, we must increase the labelSum of each vertex by at least one. As each vertex $x_i \in X$ is adjacent to exactly one vertex from both $A$ and $B$, we cannot afford to increase the label of any vertex from $A \cup B$ to dominate the vertices of $X$. This is true because we have $3q$ vertices in $X$ to dominate and only a weight of $q$. Hence, each vertex in $A \cup B$ gets a label of 1. Thus, the remaining weight available for the set $X \cup C$ is exactly $q$.

Note that each vertex \( c_j \in C \) is adjacent to exactly three vertices in \( X \) and labeling \( c_j \) with 1, will increase the labelSum of all three of its neighbors by one.
Since we have a weight of \( q \) left to consume over \( X \cup C \), and we need to increase the labelSum of all \( 3q \) vertices in \( X \) by one. The only way is to label exactly \( q \) vertices in \( C \) with 1. This ensures that all vertices in \( X \) get the required labelSum of three.
As a result, each vertex in \( X \) gets a label of 0, and exactly \( q \) vertices in \( C \) receive a label of 1. The vertices that get a label of 1 from $C$ will correspond to $\overline{C'} \subseteq \overline{C}$ that forms a solution to the X3C instance. \qed
\end{proof}
Hence, from Theorem \ref{etcresult} and Lemma \ref{lemma2}, we arrive at the following theorem.
\begin{theorem}
    R3D on split graphs is NP-complete.
\end{theorem}
\section{W[2]-hard parameterized by weight}
\begin{figure} [!ht]
    \centering
    \begin{tikzpicture} [thick,scale=0.62, every node/.style={scale=1.1}]
        \draw (3, 6) ellipse (1 and 2.1);
        \draw (3, 14) ellipse (1 and 2.1);
        \draw (3, 22) ellipse (1 and 2.1);

        \draw[blue] (3, 23.5) -- (3, 23);
        \draw[thin] (3, 23) -- (3, 22.5);
        \draw[thin] (3, 22.5) -- (3, 22);
        \draw[thin] (3, 22) -- (3, 21.5);
        \draw[thin] (3, 21.5) -- (3, 21);
        \draw[thin] (3, 21) -- (3, 20.5);

        %first
        \draw[blue] (3, 23.5) -- (7.5, 25.4);
        \draw[blue] (3, 23.5) -- (7.5, 25.1);
        \draw[thin, gray] (3, 23) -- (7.5, 25.1);
        \draw[thin, gray] (3, 23) -- (7.5, 24.8);
        \draw[thin, gray] (3, 23) -- (7.5, 25.4);
        \draw[thin, gray] (3, 22.5) -- (7.5, 24.8);
        \draw[thin, gray] (3, 22.5) -- (7.5, 24.5);
        \draw[thin, gray] (3, 22.5) -- (7.5, 25.1);
        \draw[thin, gray] (3, 22) -- (7.5, 24.5);
        \draw[thin, gray] (3, 22) -- (7.5, 24.2);
        \draw[thin, gray] (3, 22) -- (7.5, 24.8);
        \draw[thin, gray] (3, 21.5) -- (7.5, 24.2);
        \draw[thin, gray] (3, 21.5) -- (7.5, 23.9);
        \draw[thin, gray] (3, 21.5) -- (7.5, 24.5);
        \draw[thin, gray] (3, 21) -- (7.5, 23.9);
        \draw[thin, gray] (3, 21) -- (7.5, 23.6);
        \draw[thin, gray] (3, 21) -- (7.5, 24.2);
        \draw[thin, gray] (3, 20.5) -- (7.5, 23.6);
        \draw[thin, gray] (3, 20.5) -- (7.5, 23.9);

        \draw[blue] (3, 23.5) -- (7.5, 22.9);
        \draw[blue] (3, 23.5) -- (7.5, 22.6);
        \draw[thin, gray] (3, 23) -- (7.5, 22.6);
        \draw[thin, gray] (3, 23) -- (7.5, 22.3);
        \draw[thin, gray] (3, 23) -- (7.5, 22.9);
        \draw[thin, gray] (3, 22.5) -- (7.5, 22.3);
        \draw[thin, gray] (3, 22.5) -- (7.5, 22);
        \draw[thin, gray] (3, 22.5) -- (7.5, 22.6);
        \draw[thin, gray] (3, 22) -- (7.5, 22);
        \draw[thin, gray] (3, 22) -- (7.5, 21.7);
        \draw[thin, gray] (3, 22) -- (7.5, 22.3);
        \draw[thin, gray] (3, 21.5) -- (7.5, 21.7);
        \draw[thin, gray] (3, 21.5) -- (7.5, 21.4);
        \draw[thin, gray] (3, 21.5) -- (7.5, 22);
        \draw[thin, gray] (3, 21) -- (7.5, 21.4);
        \draw[thin, gray] (3, 21) -- (7.5, 21.1);
        \draw[thin, gray] (3, 21) -- (7.5, 21.7);
        \draw[thin, gray] (3, 20.5) -- (7.5, 21.1);
        \draw[thin, gray] (3, 20.5) -- (7.5, 21.4);

        \draw[blue] (3, 23.5) -- (7.5, 20.4);
        \draw[blue] (3, 23.5) -- (7.5, 20.1);
        \draw[thin, gray] (3, 23) -- (7.5, 20.1);
        \draw[thin, gray] (3, 23) -- (7.5, 20.4);
        \draw[thin, gray] (3, 23) -- (7.5, 19.8);
        \draw[thin, gray] (3, 22.5) -- (7.5, 19.8);
        \draw[thin, gray] (3, 22.5) -- (7.5, 20.1);
        \draw[thin, gray] (3, 22.5) -- (7.5, 19.5);
        \draw[thin, gray] (3, 22) -- (7.5, 19.5);
        \draw[thin, gray] (3, 22) -- (7.5, 19.2);
        \draw[thin, gray] (3, 22) -- (7.5, 19.8);
        \draw[thin, gray] (3, 21.5) -- (7.5, 19.2);
        \draw[thin, gray] (3, 21.5) -- (7.5, 18.9);
        \draw[thin, gray] (3, 21.5) -- (7.5, 19.5);
        \draw[thin, gray] (3, 21) -- (7.5, 18.9);
        \draw[thin, gray] (3, 21) -- (7.5, 18.6);
        \draw[thin, gray] (3, 21) -- (7.5, 19.2);
        \draw[thin, gray] (3, 20.5) -- (7.5, 18.6);
        \draw[thin, gray] (3, 20.5) -- (7.5, 18.9);

        %second
        \draw[red] (3, 15.5) -- (7.5, 17.4);
        \draw[red] (3, 15.5) -- (7.5, 17.1);
        \draw[thin, gray] (3, 15) -- (7.5, 17.1);
        \draw[thin, gray] (3, 15) -- (7.5, 17.4);
        \draw[thin, gray] (3, 15) -- (7.5, 16.8);
        \draw[thin, gray] (3, 14.5) -- (7.5, 16.8);
        \draw[thin, gray] (3, 14.5) -- (7.5, 17.1);
        \draw[thin, gray] (3, 14.5) -- (7.5, 16.5);
        \draw[thin, gray] (3, 14) -- (7.5, 16.5);
        \draw[thin, gray] (3, 14) -- (7.5, 16.8);
        \draw[thin, gray] (3, 14) -- (7.5, 16.2);
        \draw[thin, gray] (3, 13.5) -- (7.5, 16.5);
        \draw[thin, gray] (3, 13.5) -- (7.5, 16.2);
        \draw[thin, gray] (3, 13.5) -- (7.5, 15.9);
        \draw[thin, gray] (3, 13) -- (7.5, 16.2);
        \draw[thin, gray] (3, 13) -- (7.5, 15.9);
        \draw[thin, gray] (3, 13) -- (7.5, 15.6);
        \draw[thin, gray] (3, 12.5) -- (7.5, 15.6);
        \draw[thin, gray] (3, 12.5) -- (7.5, 15.9);

        \draw[red] (3, 15.5) -- (7.5, 14.9);
        \draw[red] (3, 15.5) -- (7.5, 14.6);
        \draw[thin, gray] (3, 15) -- (7.5, 14.6);
        \draw[thin, gray] (3, 15) -- (7.5, 14.9);
        \draw[thin, gray] (3, 15) -- (7.5, 14.3);
        \draw[thin, gray] (3, 14.5) -- (7.5, 14.3);
        \draw[thin, gray] (3, 14.5) -- (7.5, 14.6);
        \draw[thin, gray] (3, 14.5) -- (7.5, 14);
        \draw[thin, gray] (3, 14) -- (7.5, 14.3);
        \draw[thin, gray] (3, 14) -- (7.5, 14);
        \draw[thin, gray] (3, 14) -- (7.5, 13.7);
        \draw[thin, gray] (3, 13.5) -- (7.5, 14);
        \draw[thin, gray] (3, 13.5) -- (7.5, 13.7);
        \draw[thin, gray] (3, 13.5) -- (7.5, 13.4);
        \draw[thin, gray] (3, 13) -- (7.5, 13.7);
        \draw[thin, gray] (3, 13) -- (7.5, 13.4);
        \draw[thin, gray] (3, 13) -- (7.5, 13.1);
        \draw[thin, gray] (3, 12.5) -- (7.5, 13.4);
        \draw[thin, gray] (3, 12.5) -- (7.5, 13.1);

        \draw[red] (3, 15.5) -- (7.5, 12.4);
        \draw[red] (3, 15.5) -- (7.5, 12.1);
        \draw[thin, gray] (3, 15) -- (7.5, 12.4);
        \draw[thin, gray] (3, 15) -- (7.5, 12.1);
        \draw[thin, gray] (3, 15) -- (7.5, 11.8);
        \draw[thin, gray] (3, 14.5) -- (7.5, 12.1);
        \draw[thin, gray] (3, 14.5) -- (7.5, 11.8);
        \draw[thin, gray] (3, 14.5) -- (7.5, 11.5);
        \draw[thin, gray] (3, 14) -- (7.5, 11.8);
        \draw[thin, gray] (3, 14) -- (7.5, 11.5);
        \draw[thin, gray] (3, 14) -- (7.5, 11.2);
        \draw[thin, gray] (3, 13.5) -- (7.5, 11.5);
        \draw[thin, gray] (3, 13.5) -- (7.5, 11.2);
        \draw[thin, gray] (3, 13.5) -- (7.5, 10.9);
        \draw[thin, gray] (3, 13) -- (7.5, 11.2);
        \draw[thin, gray] (3, 13) -- (7.5, 10.9);
        \draw[thin, gray] (3, 13) -- (7.5, 10.6);
        \draw[thin, gray] (3, 12.5) -- (7.5, 10.6);
        \draw[thin, gray] (3, 12.5) -- (7.5, 10.9);

        \draw[dashed] (3, 23.5) to [out=210,in=165] (3, 15.5);
        \draw[blue] (3, 23.5) to [out=180,in=180] (3, 15);
        % \draw[thin, gray] (3, 23) to [out=180,in=180] (3, 15);
        % \draw[thin, gray] (3, 23) to [out=180,in=180] (3,14.5);
        \draw[red] (3, 23) to [out=210,in=150] (3,15.5);
        % \draw[thin, gray] (3, 22.5) to [out=180,in=180] (3, 14.5);
        % \draw[thin, gray] (3, 22.5) to [out=180,in=180] (3, 14);
        % \draw[thin, gray] (3, 22.5) to [out=180,in=180] (3, 15);
        % \draw[thin, gray] (3, 22) to [out=180,in=180] (3, 14);
        % \draw[thin, gray] (3, 22) to [out=180,in=180] (3, 13.5);
        % \draw[thin, gray] (3, 22) to [out=180,in=180] (3, 14.5);
        % \draw[thin, gray] (3, 21.5) to [out=180,in=180] (3, 13.5);
        % \draw[thin, gray] (3, 21.5) to [out=180,in=180] (3, 13);
        % \draw[thin, gray] (3, 21.5) to [out=180,in=180] (3, 14);
        % \draw[thin, gray] (3, 21) to [out=180,in=180] (3, 13);
        % \draw[thin, gray] (3, 21) to [out=180,in=180] (3, 12.5);
        % \draw[thin, gray] (3, 21) to [out=180,in=180] (3, 13.5);
        % \draw[thin, gray] (3, 20.5) to [out=180,in=180] (3, 12.5);
        % \draw[thin, gray] (3, 20.5) to [out=180,in=180] (3, 13);

        %%
        \draw[dashed] (3, 15.5) to [out=210,in=165] (3, 7.5);
        \draw[red] (3, 15.5) to [out=180,in=180] (3, 7);
        \draw[mynicegreen] (3, 15) to [out=210,in=150] (3, 7.5);
        % \draw[thin, gray] (3, 15) to [out=180,in=180] (3, 7);
        % \draw[thin, gray] (3, 15) to [out=180,in=180] (3,6.5);
        % \draw[thin, gray] (3, 14.5) to [out=180,in=180] (3, 7);
        % \draw[thin, gray] (3, 14.5) to [out=180,in=180] (3, 6.5);
        % \draw[thin, gray] (3, 14.5) to [out=180,in=180] (3, 6);
        % \draw[thin, gray] (3, 14) to [out=180,in=180] (3, 6.5);
        % \draw[thin, gray] (3, 14) to [out=180,in=180] (3, 6);
        % \draw[thin, gray] (3, 14) to [out=180,in=180] (3, 5.5);
        % \draw[thin, gray] (3, 13.5) to [out=180,in=180] (3, 6);
        % \draw[thin, gray] (3, 13.5) to [out=180,in=180] (3, 5.5);
        % \draw[thin, gray] (3, 13.5) to [out=180,in=180] (3, 5);
        % \draw[thin, gray] (3, 13) to [out=180,in=180] (3, 5.5);
        % \draw[thin, gray] (3, 13) to [out=180,in=180] (3, 5);
        % \draw[thin, gray] (3, 13) to [out=180,in=180] (3, 4.5);
        % \draw[thin, gray] (3, 12.5) to [out=180,in=180] (3, 5);
        % \draw[thin, gray] (3, 12.5) to [out=180,in=180] (3, 4.5);

        %%
        \draw[dashed] (3, 23.5) to [out=195,in=165] (3, 7.5);
        \draw[blue] (3, 23.5) to [out=180,in=180] (3, 7);
        \draw[mynicegreen] (3, 23) to [out=210,in=150] (3, 7.5);
        % \draw[thin, gray] (3, 23) to [out=180,in=180] (3, 7);
        % \draw[thin, gray] (3, 23) to [out=180,in=180] (3,6.5);
        % \draw[thin, gray] (3, 22.5) to [out=180,in=180] (3, 7);
        % \draw[thin, gray] (3, 22.5) to [out=180,in=180] (3, 6.5);
        % \draw[thin, gray] (3, 22.5) to [out=180,in=180] (3, 6);
        % \draw[thin, gray] (3, 22) to [out=180,in=180] (3, 6.5);
        % \draw[thin, gray] (3, 22) to [out=180,in=180] (3, 6);
        % \draw[thin, gray] (3, 22) to [out=180,in=180] (3, 5.5);
        % \draw[thin, gray] (3, 21.5) to [out=180,in=180] (3, 6);
        % \draw[thin, gray] (3, 21.5) to [out=180,in=180] (3, 5.5);
        % \draw[thin, gray] (3, 21.5) to [out=180,in=180] (3, 5);
        % \draw[thin, gray] (3, 21) to [out=180,in=180] (3, 5.5);
        % \draw[thin, gray] (3, 21) to [out=180,in=180] (3, 5);
        % \draw[thin, gray] (3, 21) to [out=180,in=180] (3, 4.5);
        % \draw[thin, gray] (3, 20.5) to [out=180,in=180] (3, 5);
        % \draw[thin, gray] (3, 20.5) to [out=180,in=180] (3, 4.5);

        %third
        \draw[mynicegreen] (3, 7.5) -- (7.5, 9.4);
        \draw[mynicegreen] (3, 7.5) -- (7.5, 9.1);
        \draw[thin, gray] (3, 7) -- (7.5, 9.4);
        \draw[thin, gray] (3, 7) -- (7.5, 9.1);
        \draw[thin, gray] (3, 7) -- (7.5, 8.8);
        \draw[thin, gray] (3, 6.5) -- (7.5, 9.1);
        \draw[thin, gray] (3, 6.5) -- (7.5, 8.8);
        \draw[thin, gray] (3, 6.5) -- (7.5, 8.5);
        \draw[thin, gray] (3, 6) -- (7.5, 8.8);
        \draw[thin, gray] (3, 6) -- (7.5, 8.5);
        \draw[thin, gray] (3, 6) -- (7.5, 8.2);
        \draw[thin, gray] (3, 5.5) -- (7.5, 8.5);
        \draw[thin, gray] (3, 5.5) -- (7.5, 8.2);
        \draw[thin, gray] (3, 5.5) -- (7.5, 7.9);
        \draw[thin, gray] (3, 5) -- (7.5, 8.2);
        \draw[thin, gray] (3, 5) -- (7.5, 7.9);
        \draw[thin, gray] (3, 5) -- (7.5, 7.6);
        \draw[thin, gray] (3, 4.5) -- (7.5, 7.9);
        \draw[thin, gray] (3, 4.5) -- (7.5, 7.6);

        \draw[mynicegreen] (3, 7.5) -- (7.5, 6.9);
        \draw[mynicegreen] (3, 7.5) -- (7.5, 6.6);
        \draw[thin, gray] (3, 7) -- (7.5, 6.9);
        \draw[thin, gray] (3, 7) -- (7.5, 6.6);
        \draw[thin, gray] (3, 7) -- (7.5, 6.3);
        \draw[thin, gray] (3, 6.5) -- (7.5, 6.6);
        \draw[thin, gray] (3, 6.5) -- (7.5, 6.3);
        \draw[thin, gray] (3, 6.5) -- (7.5, 6);
        \draw[thin, gray] (3, 6) -- (7.5, 6.3);
        \draw[thin, gray] (3, 6) -- (7.5, 6);
        \draw[thin, gray] (3, 6) -- (7.5, 5.7);
        \draw[thin, gray] (3, 5.5) -- (7.5, 6);
        \draw[thin, gray] (3, 5.5) -- (7.5, 5.7);
        \draw[thin, gray] (3, 5.5) -- (7.5, 5.4);
        \draw[thin, gray] (3, 5) -- (7.5, 5.7);
        \draw[thin, gray] (3, 5) -- (7.5, 5.4);
        \draw[thin, gray] (3, 5) -- (7.5, 5.1);
        \draw[thin, gray] (3, 4.5) -- (7.5, 5.4);
        \draw[thin, gray] (3, 4.5) -- (7.5, 5.1);

        \draw[mynicegreen] (3, 7.5) -- (7.5, 4.4);
        \draw[mynicegreen] (3, 7.5) -- (7.5, 4.1);
        \draw[thin, gray] (3, 7) -- (7.5, 4.4);
        \draw[thin, gray] (3, 7) -- (7.5, 4.1);
        \draw[thin, gray] (3, 7) -- (7.5, 3.8);
        \draw[thin, gray] (3, 6.5) -- (7.5, 4.1);
        \draw[thin, gray] (3, 6.5) -- (7.5, 3.8);
        \draw[thin, gray] (3, 6.5) -- (7.5, 3.5);
        \draw[thin, gray] (3, 6) -- (7.5, 3.8);
        \draw[thin, gray] (3, 6) -- (7.5, 3.5);
        \draw[thin, gray] (3, 6) -- (7.5, 3.2);
        \draw[thin, gray] (3, 5.5) -- (7.5, 3.5);
        \draw[thin, gray] (3, 5.5) -- (7.5, 3.2);
        \draw[thin, gray] (3, 5.5) -- (7.5, 2.9);
        \draw[thin, gray] (3, 5) -- (7.5, 3.2);
        \draw[thin, gray] (3, 5) -- (7.5, 2.9);
        \draw[thin, gray] (3, 5) -- (7.5, 2.6);
        \draw[thin, gray] (3, 4.5) -- (7.5, 2.9);
        \draw[thin, gray] (3, 4.5) -- (7.5, 2.6);

        \draw[mynicegreen] (3, 7.5) -- (3, 7);
        \draw[thin] (3, 7) -- (3, 6.5);
        \draw[thin] (3, 6.5) -- (3, 6);
        \draw[thin] (3, 6) -- (3, 5.5);
        \draw[thin] (3, 5.5) -- (3, 5);
        \draw[thin] (3, 5) -- (3, 4.5);

        \draw[red] (3, 15.5) -- (3, 15);
        \draw[thin] (3, 15) -- (3, 14.5);
        \draw[thin] (3, 14.5) -- (3, 14);
        \draw[thin] (3, 14) -- (3, 13.5);
        \draw[thin] (3, 13.5) -- (3, 13);
        \draw[thin] (3, 13) -- (3, 12.5);
        
        \filldraw (3, 14) circle (0.05);
        \filldraw (3, 14.5) circle (0.05);
        \filldraw (3, 15.5) circle (0.05); 
        \filldraw (3, 15) circle (0.05);
        \filldraw (3, 13.5) circle (0.05);
        \filldraw (3, 12.5) circle (0.05); 
        \filldraw (3, 13) circle (0.05); 

        \filldraw (3, 22) circle (0.05);
        \filldraw (3, 22.5) circle (0.05);
        \filldraw (3, 23.5) circle (0.05); 
        \filldraw (3, 23) circle (0.05);
        \filldraw (3, 21.5) circle (0.05);
        \filldraw (3, 20.5) circle (0.05); 
        \filldraw (3, 21) circle (0.05); 

        \filldraw (3, 6) circle (0.05);
        \filldraw (3, 6.5) circle (0.05);
        \filldraw (3, 7.5) circle (0.05); 
        \filldraw (3, 7) circle (0.05);
        \filldraw (3, 5.5) circle (0.05);
        \filldraw (3, 4.5) circle (0.05); 
        \filldraw (3, 5) circle (0.05); 

        %first
        \draw (7.5, 6) ellipse (0.5 and 1.1);
        
        \draw (7.5, 3.5) ellipse (0.5 and 1.1);
        
        \draw (7.5, 8.5) ellipse (0.5 and 1.1);

        \filldraw (7.5, 6) circle (0.05);
        \filldraw (7.5, 6.3) circle (0.05);
        \filldraw (7.5, 6.6) circle (0.05); 
        \filldraw (7.5, 5.7) circle (0.05);
        \filldraw (7.5, 5.4) circle (0.05);
        \filldraw (7.5, 5.1) circle (0.05); 
        \filldraw (7.5, 6.9) circle (0.05); 

        \draw[thin] (7.5, 6) -- (7.5, 6.3);
        \draw[thin] (7.5, 6.3) -- (7.5, 6.6);
        \draw[thin] (7.5, 6.6) -- (7.5, 6.9);
        \draw[thin] (7.5, 5.7) -- (7.5, 6);
        \draw[thin] (7.5, 5.4) -- (7.5, 5.7);
        \draw[thin] (7.5, 5.1) -- (7.5, 5.4);

        \filldraw (7.5, 3.5) circle (0.05);
        \filldraw (7.5, 3.8) circle (0.05);
        \filldraw (7.5, 4.1) circle (0.05); 
        \filldraw (7.5, 4.4) circle (0.05);
        \filldraw (7.5, 3.2) circle (0.05);
        \filldraw (7.5, 2.6) circle (0.05); 
        \filldraw (7.5, 2.9) circle (0.05); 

        \draw[thin] (7.5, 3.5) -- (7.5, 3.8);
        \draw[thin] (7.5, 3.8) -- (7.5, 4.1);
        \draw[thin] (7.5, 4.1) -- (7.5, 4.4);
        \draw[thin] (7.5, 3.2) -- (7.5, 3.5);
        \draw[thin] (7.5, 2.9) -- (7.5, 3.2);
        \draw[thin] (7.5, 2.6) -- (7.5, 2.9);

        \filldraw (7.5, 8.5) circle (0.05);
        \filldraw (7.5, 8.2) circle (0.05);
        \filldraw (7.5, 7.9) circle (0.05); 
        \filldraw (7.5, 7.6) circle (0.05);
        \filldraw (7.5, 8.8) circle (0.05);
        \filldraw (7.5, 9.1) circle (0.05); 
        \filldraw (7.5, 9.4) circle (0.05); 

        \draw[thin] (7.5, 8.5) -- (7.5, 8.8);
        \draw[thin] (7.5, 8.2) -- (7.5, 8.5);
        \draw[thin] (7.5, 7.9) -- (7.5, 8.2);
        \draw[thin] (7.5, 8.8) -- (7.5, 9.1);
        \draw[thin] (7.5, 9.1) -- (7.5, 9.4);
        \draw[thin] (7.5, 7.6) -- (7.5, 7.9);

        %second
        \draw (7.5, 14) ellipse (0.5 and 1.1);
        \draw (7.5, 16.5) ellipse (0.5 and 1.1);
        \draw (7.5, 11.5) ellipse (0.5 and 1.1);

        \filldraw (7.5, 14) circle (0.05);
        \filldraw (7.5, 14.3) circle (0.05);
        \filldraw (7.5, 14.6) circle (0.05); 
        \filldraw (7.5, 13.7) circle (0.05);
        \filldraw (7.5, 13.4) circle (0.05);
        \filldraw (7.5, 13.1) circle (0.05); 
        \filldraw (7.5, 14.9) circle (0.05); 

        \draw[thin] (7.5, 14) -- (7.5, 14.3);
        \draw[thin] (7.5, 14.3) -- (7.5, 14.6);
        \draw[thin] (7.5, 14.6) -- (7.5, 14.9);
        \draw[thin] (7.5, 13.7) -- (7.5, 14);
        \draw[thin] (7.5, 13.4) -- (7.5, 13.7);
        \draw[thin] (7.5, 13.1) -- (7.5, 13.4);

        \filldraw (7.5, 11.5) circle (0.05);
        \filldraw (7.5, 11.8) circle (0.05);
        \filldraw (7.5, 12.1) circle (0.05); 
        \filldraw (7.5, 12.4) circle (0.05);
        \filldraw (7.5, 11.2) circle (0.05);
        \filldraw (7.5, 10.6) circle (0.05); 
        \filldraw (7.5, 10.9) circle (0.05); 

        \draw[thin] (7.5, 11.5) -- (7.5, 11.8);
        \draw[thin] (7.5, 11.8) -- (7.5, 12.1);
        \draw[thin] (7.5, 12.1) -- (7.5, 12.4);
        \draw[thin] (7.5, 11.2) -- (7.5, 11.5);
        \draw[thin] (7.5, 10.9) -- (7.5, 11.2);
        \draw[thin] (7.5, 10.6) -- (7.5, 10.9);

        \filldraw (7.5, 16.5) circle (0.05);
        \filldraw (7.5, 16.2) circle (0.05);
        \filldraw (7.5, 15.9) circle (0.05); 
        \filldraw (7.5, 15.6) circle (0.05);
        \filldraw (7.5, 16.8) circle (0.05);
        \filldraw (7.5, 17.1) circle (0.05); 
        \filldraw (7.5, 17.4) circle (0.05); 

        \draw[thin] (7.5, 16.5) -- (7.5, 16.8);
        \draw[thin] (7.5, 16.2) -- (7.5, 16.5);
        \draw[thin] (7.5, 15.9) -- (7.5, 16.2);
        \draw[thin] (7.5, 16.8) -- (7.5, 17.1);
        \draw[thin] (7.5, 17.1) -- (7.5, 17.4);
        \draw[thin] (7.5, 15.6) -- (7.5, 15.9);

        %third
        \draw (7.5, 22) ellipse (0.5 and 1.1);
        \draw (7.5, 19.5) ellipse (0.5 and 1.1);
        \draw (7.5, 24.5) ellipse (0.5 and 1.1);

        \filldraw (7.5, 22) circle (0.05);
        \filldraw (7.5, 22.3) circle (0.05);
        \filldraw (7.5, 22.6) circle (0.05); 
        \filldraw (7.5, 21.7) circle (0.05);
        \filldraw (7.5, 21.4) circle (0.05);
        \filldraw (7.5, 21.1) circle (0.05); 
        \filldraw (7.5, 22.9) circle (0.05); 

        \draw[thin] (7.5, 22) -- (7.5, 22.3);
        \draw[thin] (7.5, 22.3) -- (7.5, 22.6);
        \draw[thin] (7.5, 22.6) -- (7.5, 22.9);
        \draw[thin] (7.5, 21.7) -- (7.5, 22);
        \draw[thin] (7.5, 21.4) -- (7.5, 21.7);
        \draw[thin] (7.5, 21.1) -- (7.5, 21.4);

        \filldraw (7.5, 19.5) circle (0.05);
        \filldraw (7.5, 19.8) circle (0.05);
        \filldraw (7.5, 20.1) circle (0.05); 
        \filldraw (7.5, 20.4) circle (0.05);
        \filldraw (7.5, 19.2) circle (0.05);
        \filldraw (7.5, 18.6) circle (0.05); 
        \filldraw (7.5, 18.9) circle (0.05); 

        \draw[thin] (7.5, 19.5) -- (7.5, 19.8);
        \draw[thin] (7.5, 19.8) -- (7.5, 20.1);
        \draw[thin] (7.5, 20.1) -- (7.5, 20.4);
        \draw[thin] (7.5, 19.2) -- (7.5, 19.5);
        \draw[thin] (7.5, 18.9) -- (7.5, 19.2);
        \draw[thin] (7.5, 18.6) -- (7.5, 18.9);

        \filldraw (7.5, 24.5) circle (0.05);
        \filldraw (7.5, 24.2) circle (0.05);
        \filldraw (7.5, 23.9) circle (0.05); 
        \filldraw (7.5, 23.6) circle (0.05);
        \filldraw (7.5, 24.8) circle (0.05);
        \filldraw (7.5, 25.1) circle (0.05); 
        \filldraw (7.5, 25.4) circle (0.05); 

        \draw[thin] (7.5, 24.5) -- (7.5, 24.8);
        \draw[thin] (7.5, 24.2) -- (7.5, 24.5);
        \draw[thin] (7.5, 23.9) -- (7.5, 24.2);
        \draw[thin] (7.5, 24.8) -- (7.5, 25.1);
        \draw[thin] (7.5, 25.1) -- (7.5, 25.4);
        \draw[thin] (7.5, 23.6) -- (7.5, 23.9);

        \filldraw (10, 6) circle (0.05);
        \filldraw (10, 3.5) circle (0.05);
        \filldraw (10, 8.5) circle (0.05);     

        \filldraw (10, 14) circle (0.05);
        \filldraw (10, 16.5) circle (0.05);
        \filldraw (10, 11.5) circle (0.05); 

        \filldraw (10, 22) circle (0.05);
        \filldraw (10, 24.5) circle (0.05);
        \filldraw (10, 19.5) circle (0.05); 

        \filldraw (12, 6) circle (0.05);
        \filldraw (12, 3.5) circle (0.05);
        \filldraw (12, 8.5) circle (0.05);     

        \filldraw (12, 14) circle (0.05);
        \filldraw (12, 16.5) circle (0.05);
        \filldraw (12, 11.5) circle (0.05); 

        \filldraw (12, 22) circle (0.05);
        \filldraw (12, 24.5) circle (0.05);
        \filldraw (12, 19.5) circle (0.05);

        \filldraw (14, 6) circle (0.05);
        \filldraw (14, 3.5) circle (0.05);
        \filldraw (14, 8.5) circle (0.05);
        
        \filldraw (14, 14) circle (0.05);
        \filldraw (14, 11.5) circle (0.05);
        \filldraw (14, 16.5) circle (0.05);
        
        \filldraw (14, 22) circle (0.05);
        \filldraw (14, 19.5) circle (0.05);
        \filldraw (14, 24.5) circle (0.05);

        \draw[thin] (7.5, 22) -- (10, 22);
        \draw[thin] (7.5, 22.3) -- (10, 22);
        \draw[thin] (7.5, 22.6) -- (10, 22);
        \draw[thin] (7.5, 21.7) -- (10, 22);
        \draw[thin] (7.5, 21.4) -- (10, 22);
        \draw[thin] (7.5, 21.1) -- (10, 22);
        \draw[thin] (7.5, 22.9) -- (10, 22);

        \draw[thin] (7.5, 19.5) -- (10, 19.5);
        \draw[thin] (7.5, 19.8) -- (10, 19.5);
        \draw[thin] (7.5, 20.1) -- (10, 19.5);
        \draw[thin] (7.5, 19.2) -- (10, 19.5);
        \draw[thin] (7.5, 18.9) -- (10, 19.5);
        \draw[thin] (7.5, 18.6) -- (10, 19.5);
        \draw[thin] (7.5, 20.4) -- (10, 19.5);
        
        \draw[thin] (7.5, 24.5) -- (10, 24.5);
        \draw[thin] (7.5, 24.2) -- (10, 24.5);
        \draw[thin] (7.5, 23.9) -- (10, 24.5);
        \draw[thin] (7.5, 24.8) -- (10, 24.5);
        \draw[thin] (7.5, 25.1) -- (10, 24.5);
        \draw[thin] (7.5, 23.6) -- (10, 24.5);
        \draw[thin] (7.5, 25.4) -- (10, 24.5);

        \draw[thin] (7.5, 14) -- (10, 14);
        \draw[thin] (7.5, 14.3) -- (10, 14);
        \draw[thin] (7.5, 14.6) -- (10, 14);
        \draw[thin] (7.5, 13.7) -- (10, 14);
        \draw[thin] (7.5, 13.4) -- (10, 14);
        \draw[thin] (7.5, 13.1) -- (10, 14);
        \draw[thin] (7.5, 14.9) -- (10, 14);

        \draw[thin] (7.5, 11.5) -- (10, 11.5);
        \draw[thin] (7.5, 11.8) -- (10, 11.5);
        \draw[thin] (7.5, 12.1) -- (10, 11.5);
        \draw[thin] (7.5, 11.2) -- (10, 11.5);
        \draw[thin] (7.5, 10.9) -- (10, 11.5);
        \draw[thin] (7.5, 10.6) -- (10, 11.5);
        \draw[thin] (7.5, 12.4) -- (10, 11.5);
        
        \draw[thin] (7.5, 16.5) -- (10, 16.5);
        \draw[thin] (7.5, 16.2) -- (10, 16.5);
        \draw[thin] (7.5, 15.9) -- (10, 16.5);
        \draw[thin] (7.5, 16.8) -- (10, 16.5);
        \draw[thin] (7.5, 17.1) -- (10, 16.5);
        \draw[thin] (7.5, 15.6) -- (10, 16.5);
        \draw[thin] (7.5, 17.4) -- (10, 16.5);

        \draw[thin] (7.5, 6) -- (10, 6);
        \draw[thin] (7.5, 6.3) -- (10, 6);
        \draw[thin] (7.5, 6.6) -- (10, 6);
        \draw[thin] (7.5, 5.7) -- (10, 6);
        \draw[thin] (7.5, 5.4) -- (10, 6);
        \draw[thin] (7.5, 5.1) -- (10, 6);
        \draw[thin] (7.5, 6.9) -- (10, 6);

        \draw[thin] (7.5, 3.5) -- (10, 3.5);
        \draw[thin] (7.5, 3.8) -- (10, 3.5);
        \draw[thin] (7.5, 4.1) -- (10, 3.5);
        \draw[thin] (7.5, 3.2) -- (10, 3.5);
        \draw[thin] (7.5, 2.9) -- (10, 3.5);
        \draw[thin] (7.5, 2.6) -- (10, 3.5);
        \draw[thin] (7.5, 4.4) -- (10, 3.5);
        
        \draw[thin] (7.5, 8.5) -- (10, 8.5);
        \draw[thin] (7.5, 8.2) -- (10, 8.5);
        \draw[thin] (7.5, 7.9) -- (10, 8.5);
        \draw[thin] (7.5, 8.8) -- (10, 8.5);
        \draw[thin] (7.5, 9.1) -- (10, 8.5);
        \draw[thin] (7.5, 7.6) -- (10, 8.5);
        \draw[thin] (7.5, 9.4) -- (10, 8.5);

        \draw[thin] (10, 16.5) -- (12, 16.5);
        \draw[thin] (10, 14) -- (12, 14);
        \draw[thin] (10,11.5) -- (12, 11.5);
        
        \draw[thin] (10, 8.5) -- (12, 8.5);
        \draw[thin] (10, 6) -- (12, 6);
        \draw[thin] (10, 3.5) -- (12, 3.5);

        \draw[thin] (10, 24.5) -- (12, 24.5);
        \draw[thin] (10, 22) -- (12, 22);
        \draw[thin] (10, 19.5) -- (12, 19.5);
        
        \draw[thin] (14, 14) -- (12, 16.5);
        \draw[thin] (14, 14) -- (12, 14);
        \draw[thin] (14, 14) -- (12, 11.5);

        \draw[thin] (14, 11.5) -- (12, 16.5);
        \draw[thin] (14, 11.5) -- (12, 14);
        \draw[thin] (14, 11.5) -- (12, 11.5);

        \draw[thin] (14, 16.5) -- (12, 16.5);
        \draw[thin] (14, 16.5) -- (12, 14);
        \draw[thin] (14, 16.5) -- (12, 11.5);
        
        \draw[thin] (14, 6) -- (12, 8.5);
        \draw[thin] (14, 6) -- (12, 6);
        \draw[thin] (14, 6) -- (12, 3.5);

        \draw[thin] (14, 3.5) -- (12, 8.5);
        \draw[thin] (14, 3.5) -- (12, 6);
        \draw[thin] (14, 3.5) -- (12, 3.5);

        \draw[thin] (14, 8.5) -- (12, 8.5);
        \draw[thin] (14, 8.5) -- (12, 6);
        \draw[thin] (14, 8.5) -- (12, 3.5);

        \draw[thin] (14, 22) -- (12, 24.5);
        \draw[thin] (14, 22) -- (12, 22);
        \draw[thin] (14, 22) -- (12, 19.5);

        \draw[thin] (14, 19.5) -- (12, 24.5);
        \draw[thin] (14, 19.5) -- (12, 22);
        \draw[thin] (14, 19.5) -- (12, 19.5);

        \draw[thin] (14, 24.5) -- (12, 24.5);
        \draw[thin] (14, 24.5) -- (12, 22);
        \draw[thin] (14, 24.5) -- (12, 19.5);

        \draw (-3.5, 3) circle (0cm) node[anchor=south]{$G$};
        \draw (7.5, 0.5) circle (0cm) node[anchor=south]{$G'$};

        \draw (3, 24.25) circle (0cm) node[anchor=south]{$A_1$};
        \draw (3, 16.25) circle (0cm) node[anchor=south]{$A_2$};
        \draw (3, 8.25) circle (0cm) node[anchor=south]{$A_3$};

        \draw (3, 23.45) circle (0cm) node[anchor=south]{$u_1$};
        \draw (3, 15.45) circle (0cm) node[anchor=south]{$u_2$};
        \draw (3, 7.45) circle (0cm) node[anchor=south]{$u_3$};

        \draw (8, 25.5) circle (0cm) node[anchor=south]{$B_1^1$};
        \draw (8, 22.8) circle (0cm) node[anchor=south]{$B_1^2$};
        \draw (8, 20.3) circle (0cm) node[anchor=south]{$B_1^3$};

        \draw (8, 17.5) circle (0cm) node[anchor=south]{$B_2^1$};
        \draw (8, 14.8) circle (0cm) node[anchor=south]{$B_2^2$};
        \draw (8, 12.3) circle (0cm) node[anchor=south]{$B_2^3$};

        \draw (8, 9.5) circle (0cm) node[anchor=south]{$B_3^1$};
        \draw (8, 6.8) circle (0cm) node[anchor=south]{$B_3^2$};
        \draw (8, 4.3) circle (0cm) node[anchor=south]{$B_3^3$};
        
        \draw (10, 24.5) circle (0cm) node[anchor=south]{$x_1^1$};
        \draw (10, 22) circle (0cm) node[anchor=south]{$x_1^2$};
        \draw (10, 19.5) circle (0cm) node[anchor=south]{$x_1^3$};
        
        \draw (10, 16.5) circle (0cm) node[anchor=south]{$x_2^1$};
        \draw (10, 14) circle (0cm) node[anchor=south]{$x_2^2$};
        \draw (10, 11.5) circle (0cm) node[anchor=south]{$x_2^3$};
        
        \draw (10, 8.5) circle (0cm) node[anchor=south]{$x_3^1$};
        \draw (10, 6) circle (0cm) node[anchor=south]{$x_3^2$};
        \draw (10, 3.5) circle (0cm) node[anchor=south]{$x_3^3$};
        
        \draw (11.75, 24.5) circle (0cm) node[anchor=south]{$y_1^1$};
        \draw (11.75, 22) circle (0cm) node[anchor=south]{$y_1^2$};
        \draw (11.75, 19.5) circle (0cm) node[anchor=south]{$y_1^3$};
        
        \draw (11.75, 16.5) circle (0cm) node[anchor=south]{$y_2^1$};
        \draw (11.75, 14) circle (0cm) node[anchor=south]{$y_2^2$};
        \draw (11.75, 11.5) circle (0cm) node[anchor=south]{$y_2^3$};
        
        \draw (11.75, 8.5) circle (0cm) node[anchor=south]{$y_3^1$};
        \draw (11.75, 6) circle (0cm) node[anchor=south]{$y_3^2$};
        \draw (11.75, 3.5) circle (0cm) node[anchor=south]{$y_3^3$};
        
        \draw (14.5, 22) circle (0cm) node[anchor=south]{$z_1^2$};
        \draw (14.5, 14) circle (0cm) node[anchor=south]{$z_2^2$};
        \draw (14.5, 6) circle (0cm) node[anchor=south]{$z_3^2$};

        \draw (14.5, 24.5) circle (0cm) node[anchor=south]{$z_1^1$};
        \draw (14.5, 16.5) circle (0cm) node[anchor=south]{$z_2^1$};
        \draw (14.5, 8.5) circle (0cm) node[anchor=south]{$z_3^1$};

        \draw (14.5, 19.5) circle (0cm) node[anchor=south]{$z_1^3$};
        \draw (14.5, 11.5) circle (0cm) node[anchor=south]{$z_2^3$};
        \draw (14.5, 3.5) circle (0cm) node[anchor=south]{$z_3^3$};

        \filldraw (-3.5, 6) circle (0.05);
        \filldraw (-3.5, 6.5) circle (0.05);
        \filldraw (-3.5, 7.5) circle (0.05); 
        \filldraw (-3.5, 7) circle (0.05);
        \filldraw (-3.5, 5.5) circle (0.05);
        \filldraw (-3.5, 4.5) circle (0.05); 
        \filldraw (-3.5, 5) circle (0.05); 

        \draw[thin] (-3.5, 7.5) -- (-3.5, 7);
        \draw[thin] (-3.5, 7) -- (-3.5, 6.5);
        \draw[thin] (-3.5, 6.5) -- (-3.5, 6);
        \draw[thin] (-3.5, 6) -- (-3.5, 5.5);
        \draw[thin] (-3.5, 5.5) -- (-3.5, 5);
        \draw[thin] (-3.5, 5) -- (-3.5, 4.5);
        
        \draw (-3.5, 7.6) circle (0cm) node[anchor=south]{$u$};
        
        % \draw (7, 5.5) circle (0cm) node[anchor=south]{clique};
    \end{tikzpicture}
    \caption{Graph $G'$ constructed from $G$ and $k=3$. In the induced subgraph $G[A_1\cup A_2\cup A_3]$, only the edges incident on $u_1$, $u_2$ and $u_3$ are shown in the figure. Dashed edges indicate the edges between the vertices of the set $\{u_1, u_2, u_3\}$. The edges incident on $u_1$, $u_2$ and $u_3$ (except dashed) are colored in blue, red and green, respectively. }
    \label{fig:fig3}
\end{figure}
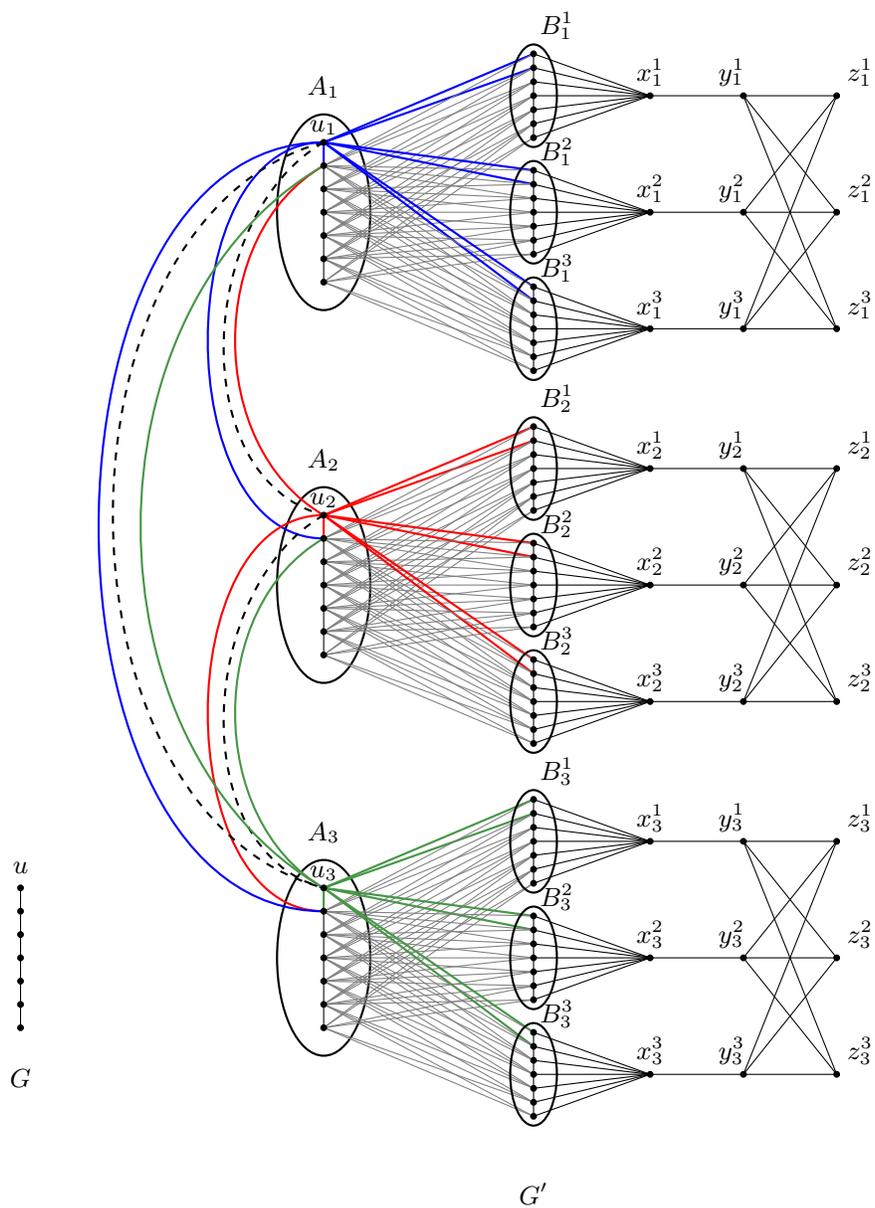

It was asked in~\cite{CHAUDHARY2024301} regarding the parameterized complexity of R3D. In this section, we investigate the parameterized complexity of R3D and prove that the problem is W[2]-hard parameterized by weight. We achieve this result by providing a parameterized reduction from \ds{} which is known to be W[2]-hard parameterized by solution size~\cite[Theorem 13.21]{downey1999parameterized}. Given a graph $G = (V, E)$, a set $S \subseteq V$ is a dominating set if every vertex $u \in V\setminus S$ has a neighbour in $S$. The decision version of \ds{} (DS) asks whether there exists a dominating set of size at most $k$.

Given an instance \(I = (G, k)\) of DS, we construct an instance \(I' = (G', 12k)\) of R3D as follows. Without loss of generality, we assume that $k$ in DS instance $I$ is a multiple of 3 (as this can be achieved by adding the isolated vertices to $G$).
\begin{itemize}
    \item We begin by creating three disjoint copies of \(G\), denoted by \(A_1\), \(A_2\) and \(A_3\), and include them in \(G'\). Vertex $u \in V(G)$ is denoted by $u_1$ in $A_1$, $u_2$ in $A_2$ and $u_3$ in $A_3$.
    \item Additionally, for every \(i \in \{1, 2, 3\}\) and \(j \in \{1, 2, \dots, k\}\), we create a copy of the vertex set \(V(G)\), denoted \(B_i^j\). Thus, there are \(3k\) such sets in total. Vertex $v \in V(G)$ is denoted by $v_i^j$ in $B_i^j$, for each \(i \in \{1, 2, 3\}\) and \(j \in \{1, 2, \dots, k\}\). 
    \item For each set \(B_i^j\), we introduce three new vertices \(x_i^j\), \(y_i^j\) and \(z_i^j\) in \(G'\).
    \item For every \(i \in \{1, 2, 3\}\) and \(j \in \{1, 2, \dots, k\}\), we make \(x_i^j\) adjacent to \(y_i^j\) and also we make \(x_i^j\) adjacent to each vertex of \(B_i^j\). 
    \item We make \(y_i^{3j+1}\), \(y_i^{3j+2}\) and \(y_i^{3j+3}\) adjacent to each vertex among \(z_i^{3j+1}\), \(z_i^{3j+2}\) and \(z_i^{3j+3}\), for every \(i \in \{1, 2, 3\}\) and $j \in \{0\} \cup [\frac{k}{3}-1]$.
    \item For every \(i \in \{1, 2, 3\}\) and \(j \in \{1, 2, \dots, k\}\), and for every pair of vertices \(u_i \in A_i\) and \(v_i^j \in B_i^j\), we add an edge between \(u_i\) and \(v_i^j\) if and only if \(u \in N_G(v)\). 
    \item Similarly, for every \(i \in \{1, 2, 3\}\) and \(j \in \{1, 2, \dots, k\}\), and for every pair of vertices \(u_i \in A_i\) and \(u_i^j \in B_i^j\), we add an edge between \(u_i\) and \(u_i^j\).
    \item For every \( i, j \in \{1, 2, 3\} \), and for each pair of vertices \( u_i \in A_i \) and \( v_j \in A_j \), we add an edge between \( u_i \) and \( v_j \) if and only if \( u \in N_G(v) \).
    \item Finally, for every \( i, j \in \{1, 2, 3\} \) with $i \neq j$, and for each pair of vertices \( u_i \in A_i \) and \( u_j \in A_j \), we add an edge between \( u_i \) and \( u_j \).
\end{itemize}  
This completes the construction of \(G'\). See Fig. \ref{fig:fig3} for more details. The size of $G'$ is $\mathcal{O}(n^2)$.
\begin{lemma} \label{lemma3}
    There exists a dominating set of size at most $k$ if and only if there exists a Roman \{3\}-dominating function of weight at most $12k$.
\end{lemma}
\begin{proof}
    $[\Rightarrow]$ Let $S$ be a dominating set of size at most $k$. We define a Roman \{3\}-dominating function $f: V \rightarrow \{0,1,2,3\}$ as follows. 
    \[
f(v) = 
\begin{cases}
0, & \text{if } v \in \bigcup_{u \notin S} \{u_1, u_2, u_3\} \cup \bigcup_{i \in [3]} \bigcup_{j\in[k]} (B_i^j \cup \{z_i^j\}), \\
1, & \text{if } v \in \bigcup_{u \in S} \{u_1, u_2, u_3\} \cup \bigcup_{i \in [3]} \bigcup_{j\in[k]} y_i^j, \\
2, & \text{if } v \in \bigcup_{i \in [3]} \bigcup_{j\in[k]} x_i^j, \\
3, & \text{if } v \in \emptyset
\end{cases}
\]
\begin{itemize}
    \item Every vertex in $\bigcup_{i \in [3]}\bigcup_{j\in [k]}\{y_i^j, z_i^j\}$ has a labelSum of three from its closed neighbours in $\bigcup_{i \in [3]}\bigcup_{j\in [k]}\{x_i^j, y_i^j\}$.
    \item Every vertex in $\bigcup_{i \in [3]}\bigcup_{j\in [k]} B_i^j$ has a labelSum of three from its neighbour $x_i^j$ and its neighbours in $A_i$.
    \item Every vertex in $A_1 \cup A_2 \cup A_3$ has a labelSum of three from its three closed neighbours with a label of 1 in $A_1 \cup A_2 \cup A_3$ itself.
\end{itemize}
The weight of $\bigcup_{i \in [3]}\bigcup_{j\in [k]}\{x_i^j, y_i^j\}$ is $9k$ and the weight of $A_1 \cup A_2 \cup A_3$ is at most $3k$. Hence, we conclude that $f$ is a Roman \{3\}-dominating function for $G'$ with weight at most $12k$. 

\medskip \noindent
$[\Leftarrow]$ Let $f$ be a Roman \{3\}-dominating function of $G'$ with weight at most $12k$. 
\begin{itemize}
    \item As each $z_i^j$ is only adjacent to three vertices from $\bigcup_{i \in [3]}\bigcup_{j\in [k]} y_i^j$ and no other vertices, we must label each vertex in $\bigcup_{i \in [3]}\bigcup_{j\in [k]} z_i^j$ with 0, as the neighbours from $\bigcup_{i \in [3]}\bigcup_{j\in [k]} y_i^j$ will get a positive label.
    \item For each $i \in [3]$ and $j \in [k]$, $y_i^j$ must be labeled at least 1, in order to dominate its neighbours in $\bigcup_{i \in [3]}\bigcup_{j\in [k]}$ $z_i$.
    \item As the vertex $x_i^j$ is adjacent to each vertex of the set $B_i^j$, we label $x_i^j$ with 2 and fix the label of 1 to $y_i^j$. After assigning the labels to $x_i^j$ and $y_i^j$, the labelSum of $y_i^j$ becomes three and each vertex of the set $B_i^j$ has a labelsum of two.
    \item At this point, each vertex of the set $\bigcup_{i \in [3]}\bigcup_{j\in [k]}\{x_i^j, y_i^j, z_i^j\}$ has a labelSum of three.
    \item Every vertex in $\bigcup_{i \in [3]}\bigcup_{j\in [k]} B_i^j$ has a labelSum of two from its neighbour $x_i^j$. So far, the vertices of $A$ have a labelSum of zero.
    \item The weight of $\bigcup_{i \in [3]}\bigcup_{j\in [k]}\{x_i^j, y_i^j, z_i^j\}$ is $9k$. At this point, we are left with weight at most $3k$.
    \item We label $k$ vertices from each of $A_1$, $A_2$ and $A_3$ with a label of 1 such that each vertex in $\bigcup_{i \in [3]}\bigcup_{j\in [k]} B_i^j$ is adjacent to at least one of them.
    \item Similarly, every vertex in \( A_1 \cup A_2 \cup A_3 \) must have a labelSum of three from its closed neighborhood within \( A_1 \cup A_2 \cup A_3 \). This condition is satisfied only if the \( k \) vertices that are labeled 1 in each of \( A_1 \), \( A_2 \), and \( A_3 \) are adjacent to all the other vertices within their respective sets.
    \item The \( k \) vertices that are labeled 1 in \( A_1 \) (or \( A_2 \), or \( A_3 \)) form a dominating set \( S \) of size \( k \) in \( G \).
\end{itemize}\qed
\end{proof}
Hence, from Lemma \ref{lemma3}, we obtain the following result.
\begin{theorem}
    R3D is W[2]-hard parameterized by weight.
\end{theorem}
\section{Polynomial-time algorithm for block graphs}
The complexity of R3D on block graphs was posed as an open question in~\cite{CHAUDHARY2024301}. In this section, we study the complexity of R3D on block graphs and present an algorithm that runs in $\mathcal{O}(n^3)$ time.\vspace{2mm} \\
Given a graph $G=(V, E)$ and let $R(G)$ be a set of functions (not necessarily Roman \{3\}-dominating) on $G$. For a vertex $u \in V(G)$ and $i \in \{0\} \cup [8]$, we define the sets $D^i(G,u)$ as follows. 
\begin{enumerate}
    \item $D^i(G, u)=\{f \in R(G):f$ is a Roman \{3\}-dominating function on $G$ such that $f(u) = i\}$, for each $i \in \{0\} \cup [3]$.
    \item $D^4(G, u)=\{f \in R(G): f_{G-u}$ is a Roman \{3\}-dominating function on $G-u$, either there exists exactly one vertex $v \in N_G(u)$ with $f(v) =2$ and for each $w \in N_G[u] \setminus \{v\}$\}, $f(w) = 0$ or there exist exactly two vertices $v_1, v_2 \in N_G(u)$ with $f(v_1)= f(v_2) = 1$ and for each $w \in N_G[u] \setminus \{v_1, v_2\}$\}, $f(w) = 0$.
    \item $D^5(G, u)=\{f \in R(G): f_{G-u}$ is a Roman \{3\}-dominating function on $G-u$, there exists exactly one vertex $v \in N_G(u)$ with $f(v) =1$ and $f(w) = 0$ for each $w \in N_G[u] \setminus \{v\}$\}.
    \item $D^6(G, u)=\{f \in R(G): f_{G-u}$ is a Roman \{3\}-dominating function on $G-u$, $f(w) = 0$ for each $w \in N_G[u]$\}. 
    \item $D^7(G, u)=\{f \in R(G): f_{G-u}$ is a Roman \{3\}-dominating function on $G-u$, $f(u) = 1$ and there exists exactly one vertex $v \in N(u)$ with $f(v) = 1$ and $f(w) = 0$ for each $w \in N_G(u) \setminus \{v\}$\}.
    \item $D^8(G, u)=\{f \in R(G): f_{G-u}$ is a Roman \{3\}-dominating function on $G-u$, $f(u)=1$ and $f(w) = 0$ for each $w \in N_G(u)$\}.
\end{enumerate}

\noindent 
% $D^1(G, u)=\{f: f$ is a Roman \{3\}-dominating function on $G$ such that $f(u) = 1\}$.\vspace{2mm} \\
% $D^2(G, u)=\{f: f$ is a Roman \{3\}-dominating function on $G$ such that $f(u) = 2\}$.\vspace{2mm} \\
% $D^3(G, u)=\{f: f$ is a Roman \{3\}-dominating function on $G$ such that $f(u) = 3\}$.\vspace{2mm} \\
% (2) $D^4(G, u)=\{f \in R(G): f_{G-u}$ is a Roman \{3\}-dominating function on $G-u$, there exists exactly one vertex $v \in N(u)$ with $f(v) =2$ or exactly two vertices $v_1, v_2 \in N(u)$ with $f(v_1)= f(v_2) = 1$ and for each $w \in N_G[u] \setminus \{v, v_1, v_2\}$\}, $f(w) = 0$.

% \noindent 
% (3) $D^5(G, u)=\{f \in R(G): f_{G-u}$ is a Roman \{3\}-dominating function on $G-u$, there exists exactly one vertex $v \in N(u)$ with $f(v) =1$ and $f(w) = 0$ for each $w \in N_G[u] \setminus \{v\}$\}.

% \noindent
% (4) $D^6(G, u)=\{f \in R(G): f_{G-u}$ is a Roman \{3\}-dominating function on $G-u$, $f(v) = 0$ for each $v \in N_G[u]$\}. \vspace{2mm} \\
\noindent Let $\gamma_{R3}^i(G, u) = \min(w(f): f\in D^i(G, u))$ for each $i \in \{0
\} \cup [8]$.
\begin{lemma}~\label{lemma4}
    Consider a graph $G$ and let $v$ be an arbitrary vertex of $G$. Then,
    $\gamma_{R3}(G)$ = $\min(\gamma_{R3}^0(G, v), \gamma_{R3}^1(G, v), \gamma_{R3}^2(G, v),\gamma_{R3}^3(G, v))$.
\end{lemma} 
\begin{proof}
    Let $v$ be an arbitrary vertex of $G$. The value of $\gamma_{R3}^i(G, v)$ indicate the minimum weight of a Roman \{3\}-dominating function under the constraint that $f(v) = i$. Since the label assigned to any vertex $v \in V(G)$ can be one among four values of the set \{0,1,2,3\}, it follows that $\gamma_{R3}(G)$ can be obtained by computing the minimum value among $\gamma_{R3}^0(G, v)$, $\gamma_{R3}^1(G, v)$, $\gamma_{R3}^2(G, v)$ and $\gamma_{R3}^3(G, v)$. \qed
\end{proof}
\begin{lemma}~\label{lemma5a}
    Let $f$ be a Roman \{3\}-dominating function on $G$ and let $G'$ be an induced subgraph of $G$ such that there exists a vertex $v \in V(G')$ such that $N_G(v) \setminus N_{G'}(v) \neq \emptyset$ and $N_{G}(u) = N_{G'}(u)$ for each vertex $u \in V(G') \setminus \{v\}$. If $f(v) = 0$, then the following hold: 
    \begin{enumerate} [(i)]
        \item If there exists exactly one vertex $v_a \in N_G(v) \setminus N_{G'}(v)$ such that $f(v_a) = 1$ and $f(u) = 0$ for each vertex $u \in N_G(v) \setminus (N_{G'}(v)\cup \{v_a\})$, then $f_{G'} \in D^0(G',v) \cup D^4(G',v)$.
        \item If there exists exactly one vertex $v_a \in N_G(v) \setminus N_{G'}(v)$ such that $f(v_a) = 2$ and $f(u) = 0$ for each vertex $u \in N_G(v) \setminus (N_{G'}(v)\cup \{v_a\})$, then $f_{G'} \in D^0(G',v) \cup D^4(G',v)\cup D^5(G', v)$.
        \item If there exist exactly two vertices $v_a,v_b \in N_G(v) \setminus N_{G'}(v)$ such that $f(v_a) = f(v_b) = 1$ and $f(u) = 0$ for each vertex $u \in N_G(v) \setminus (N_{G'}(v)\cup \{v_a, v_b\})$, then $f_{G'} \in D^0(G',v) \cup D^4(G',v)\cup D^5(G', v)$.
        \item If there exist three vertices $v_a, v_b, v_c \in N_G(v) \setminus N_{G'}(v)$ such that $f(v_a) \geq 1, f(v_b) \geq 1, f(v_c) \geq 1$ or there exist two vertices $v_a,v_b \in N_G(v) \setminus N_{G'}(v)$ such that $f(v_a) \geq 2$ and $f(v_b) \geq 1$ or there exists a vertex $v_a \in N_G(v) \setminus N_{G'}(v)$ such that $f(v_a) =3$, then $f_{G'} \in D^0(G',v) \cup D^4(G',v)\cup D^5(G', v)\cup D^6(G', v)$.
    \end{enumerate}
\end{lemma}
\begin{proof}
Let $f(v)$ = 0. We consider the case where either
\begin{itemize}
    \item there exists a vertex $v_p \in N_{G'}(v)$ such that $f(v_p) = 3$ or
    \item there exist two vertices $v_p, v_q\in N_{G'}(v)$ such that $f(v_p) \geq 2$ and $f(v_q) \geq 1$ or
    \item there exist three vertices $v_p, v_q, v_r\in N_{G'}(v)$ such that $f(v_p) \geq 1$, $f(v_q) \geq 1$ and $f(v_r) \geq 1$.
\end{itemize} Since there exist vertices from $G'$ with the sum of their labels to be at least three, and $G'$ is a complete graph, it follows that every vertex in $G'$ is dominated. Therefore, irrespective of the labels of the neighbours of $v$ in $N_G(v)\setminus N_{G'}(v)$, we have that $f_{G'} \in D^0(G',v)$.

From now on, we assume that the sum of the labels of the vertices in $G'$ is at most two. We have the following cases based on the labels of the vertices in $N_G(v)\setminus N_{G'}(v)$. \vspace{2mm} \\
(i) Assume that there exists exactly one vertex $v_a \in N_G(v)\setminus N_{G'}(v)$ such that $f(v_a) =1$ and $f(u) = 0$ for each vertex $u \in N_G(v)\setminus (N_{G'}(v) \cup \{v_a\})$. By the definition of Roman \{3\}-dominating function $v$ must have an additional labelSum of at least two from $N_{G'}(v)$. Therefore, there must exist a vertex with a label of at least 2 or two vertices each with a label of at least 1 in $N_{G'}(v)$. Hence, we obtain that $f_{G'} \in D^0(G',v) \cup D^4(G', v)$.\vspace{2mm} \\
(ii) Assume that there exists exactly one vertex $v_a \in N_G(v)\setminus N_{G'}(v)$ such that $f(v_a) =2$ and $f(u) = 0$ for each vertex $u \in N_G(v)\setminus (N_{G'}(v) \cup \{v_a\})$. By the definition of Roman \{3\}-dominating function $v$ must have an additional labelSum of at least one from $N_{G'}(v)$. Therefore, there must exist a vertex with a label of at least 1 in $N_{G'}(v)$. Hence, we obtain that $f_{G'} \in D^0(G',v) \cup D^4(G', v) \cup D^5(G', v)$.\vspace{2mm} \\
(iii) Assume that there exists exactly two vertices $v_a, v_b \in N_G(v)\setminus N_{G'}(v)$ such that $f(v_a) = f(v_b) = 1$ and $f(u) = 0$ for each vertex $u \in N_G(v)\setminus (N_{G'}(v) \cup \{v_a, v_b\})$. By the definition of Roman \{3\}-dominating function $v$ must have an additional labelSum of at least one from $N_{G'}(v)$. There must exist a vertex with a label of at least 1 in $N_{G'}(v)$. Hence, we obtain that $f_{G'} \in D^0(G',v) \cup D^4(G', v) \cup D^5(G, v)$.\vspace{2mm} \\
(iv) Assume that there exists either one vertex $v_a \in N_G(v)\setminus N_{G'}(v)$ such that $f(v_a) =3$ or there exist two vertices $v_a, v_b \in N_G(v)\setminus N_{G'}(v)$ such that $f(v_a) \geq 2$ and $f(v_b) \geq 1$ or there exist three vertices $v_a, v_b, v_c \in N_G(v)\setminus N_{G'}(v)$ such that $f(v_a) \geq 1$, $f(v_b) \geq 1$ and $f(v_c) \geq 1$. It is to be noted that the labelSum of $v$ from $N_G(v)\setminus N_{G'}(v)$ is at least three. Hence, we obtain that $f_{G'} \in D^0(G',v) \cup D^4(G', v) \cup D^5(G',v) \cup D^6(G', v)$. \qed
\end{proof}
\begin{lemma}~\label{lemma5b}
    Let $f$ be a Roman \{3\}-dominating function on $G$ and let $G'$ be an induced subgraph of $G$ such that there exists a vertex $v \in V(G')$ such that $N_G(v) \setminus N_{G'}(v) \neq \emptyset$ and $N_{G}(u) = N_{G'}(u)$ for each vertex $u \in V(G') \setminus \{v\}$. If $f(v) = 1$, then the following hold: 
    \begin{enumerate} [(i)]
        \item If there exists exactly one vertex $v_a \in N_G(v) \setminus N_{G'}(v)$ such that $f(v_a) = 1$ and $f(u) = 0$ for each vertex $u \in N_G(v) \setminus (N_{G'}(v)\cup \{v_a\})$, then $f_{G'} \in D^1(G',v) \cup D^7(G',v)$.
        \item If there exists a vertex $v_a \in N_G(v) \setminus N_{G'}(v)$ such that $f(v_a) = 2$ or there exist two vertices $v_a,v_b \in N_G(v) \setminus N_{G'}(v)$ such that $f(v_a) = f(v_b) = 1$, then $f_{G'} \in D^1(G',v) \cup D^7(G', v) \cup D^8(G', v)$.
    \end{enumerate}
\end{lemma}
\begin{proof}
    Let $f(v)$ = 1. We consider the case where either
\begin{itemize}
    \item there exists a vertex $v_p \in N_{G'}(v)$ such that $f(v_p) \geq 2$ or
    \item there exist two vertices $v_p, v_q\in N_{G'}(v)$ such that $f(v_p) \geq 1$ and $f(v_q) \geq 1$
\end{itemize} Since there exist vertices from $G'$ with the sum of their labels to be at least two, and $G'$ is a complete graph, it follows that every vertex in $G'$ is dominated. Therefore, irrespective of the labels of the neighbours of $v$ in $N_G(v)\setminus N_{G'}(v)$, we have that $f_{G'} \in D^1(G',v)$.

From now on, we assume that the sum of the labels of the vertices in $G'$ is at most one. We have the following cases based on the labels of the vertices in $N_G(v)\setminus N_{G'}(v)$. \vspace{2mm} \\
(i) Assume that there exists exactly one vertex $v_a \in N_G(v)\setminus N_{G'}(v)$ such that $f(v_a) =1$ and $f(u) = 0$ for each $u \in N_G(v)\setminus (N_{G'}(v) \cup \{v_a\})$. As $f(v) = 1$, by the definition of Roman \{3\}-dominating function $v$ must have an additional labelSum of at least one from $N_{G'}(v)$. Therefore, there must exist a vertex with a label of at least 1 in $N_{G'}(v)$. Hence, we obtain that $f_{G'} \in D^1(G',v) \cup D^7(G', v)$.\vspace{2mm} \\
(ii) Assume that there exists a vertex $v_a \in N_G(v)\setminus N_{G'}(v)$ such that $f(v_a) =2$ or there exist two vertices $v_a, v_b \in N_G(v)\setminus N_{G'}(v)$ such that $f(v_a) \geq 1$ and $f(v_b) \geq 1$. It is to be noted that the labelSum of $v$ from $N_G(v)\setminus N_{G'}(v)$ is at least two and $f(v) = 1$. Hence, we obtain that $f_{G'} \in D^1(G',v) \cup D^7(G',v) \cup D^8(G', v)$. \qed
\end{proof}
\begin{observation} ~\label{obs1}
    Let $f$ be a Roman \{3\}-dominating function on $G$ and let $G'$ be an induced subgraph of $G$ such that there exists a vertex $v \in V(G')$ such that $N_G(v) \setminus N_{G'}(v) \neq \emptyset$ and $N_{G}(u) = N_{G'}(u)$ for each $u \in V(G') \setminus \{v\}$. If $f(v) = 2$ $($resp. $f(v) = 3)$, then $f_{G'} \in D^2(G',v)$ $($resp. $f_{G'} \in D^3(G',v))$.
\end{observation}
\begin{figure} [t]
    \centering
    \begin{tikzpicture} [thick,scale=0.7, every node/.style={scale=0.82}]
        \draw (7, 5) ellipse (5 and 2);
        \draw (2.5, 4) [rotate=-15] ellipse (1 and 2);
        \draw (6, 3) ellipse (1 and 2);
        \draw (11, 0.5) [rotate=15] ellipse (1 and 2);

        \filldraw[gray] (8, 4.25) circle (0.05);
        \filldraw[gray] (7.8, 4.25) circle (0.05);
        \filldraw[gray] (8.2, 4.25) circle (0.05);     

        \filldraw[gray] (8, 2.25) circle (0.05);
        \filldraw[gray] (7.8, 2.25) circle (0.05);
        \filldraw[gray] (8.2, 2.25) circle (0.05); 

        \filldraw[black] (3.85, 4.5) circle (0.1);
        \filldraw[black] (6, 4.25) circle (0.1);
        \filldraw[black] (10.2, 4.6) circle (0.1); 
        % \draw[thin, dashed] (2.3, 3.7) -- (7.73, 3.7);
        
        \draw (3.55, 4.05) circle (0cm) node[anchor=south]{$v_1$};
        \draw (5.65, 3.75) circle (0cm) node[anchor=south]{$v_2$};
        \draw (10, 4) circle (0cm) node[anchor=south]{$v_k$};

        \draw (3.25, 1.95) circle (0cm) node[anchor=south]{$H_1$};
        \draw (5.95, 1.6) circle (0cm) node[anchor=south]{$H_2$};
        \draw (10.75, 2.1) circle (0cm) node[anchor=south]{$H_k$};
        \draw (7, 5.5) circle (0cm) node[anchor=south]{clique};
    \end{tikzpicture} \vspace{-3mm}
    \caption{Graph $H$}
    \label{fig:fig4} \vspace{-4mm}
\end{figure}
To design our algorithm, we assume that each graph is rooted at a designated vertex. The following lemma forms the foundation of our polynomial-time algorithm for block graphs. Our approach relies on a specific graph composition technique.
Let \( H_1, H_2, \ldots, H_k \) (with \( k \geq 2 \)) be the graphs rooted at vertices \( v_1, v_2, \ldots, v_k \), respectively. We define a graph \( H \) to be the \emph{composition} of \( H_1, H_2, \ldots, H_k \) if it is obtained by taking their disjoint union and then adding edges to make the set \( \{v_1, v_2, \ldots, v_k\} \) a clique in \( H \).
Note that by starting with trivial graphs and repeatedly applying this composition operation, we can construct any connected block graph (refer to~\cite{chain1996weighted}). See Fig. \ref{fig:fig4} for more details.

Given a graph $H$ rooted at $v_1$, we compute the values of $\gamma_{R3}^i(H, v_1)$, for $i \in \{0\} \cup [8]$ as follows.
\begin{lemma} ~\label{lemma6}
    Let $H_1, H_2, ..., H_k (k \geq 2)$ be the graphs rooted at $v_1, v_2, ..., v_k$, respectively and $H$ be the graph rooted at $v_1$. We obtain the values of $\gamma_{R3}^i(H, v_1)$, for $i \in \{0\} \cup [8]$ as follows. \vspace{2mm} \\
    $(a) \hspace{2mm} \gamma_{R3}^0(H, v_1) =$     
    $\min\begin{cases}
        \sum_{1 \leq r \leq k} \gamma_{R3}^0 (H_r, v_r); \\
        \min \{\gamma_{R3}^i (H_1, v_1):i =0, 4 \}+C_1; \\
        \min \{\gamma_{R3}^i (H_1, v_1):i =0, 4, 5\}+\min\{B_1, B_2\}; \\
        \min \{\gamma_{R3}^i (H_1, v_1):i =0, 4, 5, 6\}+\min\{A_1,A_2,A_3\}. \\
    \end{cases}$ \\
    $(b) \hspace{2mm} \gamma_{R3}^1(H, v_1)=$
        $\min\begin{cases}
        \gamma_{R3}^1 (H_1, v_1)+\sum_{r \in \{2, ..., k\}}\min\{\gamma_{R3}^i(H_r, v_r):i=0,4\};\\
        \min \{\gamma_{R3}^i (H_1, v_1):i =1, 7\}+C_2; \\
        % \gamma_{R3}^1 (H_1, v_1)+\min\{C_3, D_1\}; \\
        \min \{\gamma_{R3}^i (H_1, v_1):i =1, 7, 8\}+\min\{B_3,B_4\}. 
    \end{cases}$  \\
    $(c) \hspace{2mm} \gamma_{R3}^2(H, v_1)=$  
        $\min\begin{cases}
        \gamma_{R3}^2 (H_1, v_1)+\sum_{r \in \{2, ..., k\}}\min\{\gamma_{R3}^i(H_r, v_r):i=0,4,5\}; \\
        \gamma_{R3}^2 (H_1, v_1)+C_3; \\
    \end{cases}$ \\
    $(d) \hspace{2mm} \gamma_{R3}^3(H, v_1) = \gamma_{R3}^3 (H_1, v_1)+\sum_{r \in \{2, ..., k\}}\min\{\gamma_{R3}^i(H_r, v_r):i=0,1,2,3,4,5,6,7,8\}.$\\
    $(e) \hspace{2mm} \gamma_{R3}^4(H, v_1) =$     $\min\begin{cases}
        \gamma_{R3}^4 (H_1, v_1)+\sum_{r \in \{2, ..., k\}}\gamma_{R3}^0(H_r, v_r); \\
        \gamma_{R3}^5 (H_1, v_1)+C_1;\\
        \gamma_{R3}^6 (H_1, v_1)+\min\{B_1,B_2\}.
    \end{cases}$ \\
    $(f) \hspace{2mm} \gamma_{R3}^5(H, v_1) =$     $\min\begin{cases}
        \gamma_{R3}^5 (H_1, v_1)+\sum_{r \in \{2, ..., k\}}\gamma_{R3}^0(H_r, v_r); \\
        \gamma_{R3}^6 (H_1, v_1)+C_1.\\
    \end{cases}$ \\
    $(g) \hspace{2mm} \gamma_{R3}^6(H, v_1) = \gamma_{R3}^6 (H_1, v_1)+\sum_{r \in \{2, ..., k\}}\gamma_{R3}^0(H_r, v_r)$. \vspace{2mm} \\
    $(h) \hspace{2mm} \gamma_{R3}^7(H, v_1) =$     $\min\begin{cases}
        \gamma_{R3}^7 (H_1, v_1)+\sum_{r \in \{2, ..., k\}}\gamma_{R3}^0(H_r, v_r); \\
        \gamma_{R3}^8 (H_1, v_1)+C_2.\\
    \end{cases}$ \\
    $(i) \hspace{2mm} \gamma_{R3}^8(H, v_1) = \gamma_{R3}^8 (H_1, v_1)+\sum_{r \in \{2, ..., k\}}\gamma_{R3}^0(H_r, v_r)$. \vspace{2mm} \\
    where 
    $A_1 = \min_{a \in \{2,...,k\}}\{\gamma_{R3}^3 (H_a, v_a)+\sum_{r \in \{2, ..., k\}\setminus\{a\}}\min\{\gamma_{R3}^i(H_r, v_r):i=0,1,2,3,4,5,6,7,8\}\}$, \\
    % $A_2$ = $\min_{a,b \in \{2,...,k\}}\{\gamma_{R3}^2 (H_a, v_a)+\gamma_{R3}^2 (H_b, v_b)+\sum_{r \in \{2, ..., k\}\setminus\{a,b\}}\min\{\gamma_{R3}^i(H_r, v_r):i=0,1,2,3,4,5,6,7,8\}\}$, \\
    $A_2 = \min_{a,b \in \{2,...,k\}}\{\gamma_{R3}^2 (H_a, v_a)+\min\{\gamma_{R3}^i (H_b, v_b): i=1,7,8\}+\sum_{r \in \{2, ..., k\}\setminus\{a,b\}}$ $\min\{\gamma_{R3}^i(H_r, v_r):i=0,1,2,3,4,5,6,7,8\}\}$, \\
    $A_3 = \min_{a,b,c \in \{2,...,k\}}\{\min\{\gamma_{R3}^i (H_a, v_a): i=1,7,8\}+\min\{\gamma_{R3}^i (H_b, v_b): i=1,7,8\}+\min\{\gamma_{R3}^i (H_c, v_c): i=1,7,8\}+\sum_{r \in \{2, ..., k\}\setminus\{a,b,c\}}\min\{\gamma_{R3}^i(H_r, v_r):i=0,1,2,3,4,5,6,7,8\}\}$, \\
    $B_1 = \min_{a \in \{2,...,k\}}\{\gamma_{R3}^2 (H_a, v_a)+\sum_{r \in \{2, ..., k\}\setminus\{a\}}\min\{\gamma_{R3}^i(H_r, v_r):i=0,1,4,5,7,8\}\}$, \\
    $B_2 = \min_{a,b \in \{2,...,k\}}\{\min\{\gamma_{R3}^i (H_a, v_a): i=1,7\}+\min\{\gamma_{R3}^i (H_b, v_b): i=1,7\}+\sum_{r \in \{2, ..., k\}\setminus\{a,b\}}\min\{\gamma_{R3}^i(H_r, v_r):i=0,1,4,5,7,8\}\}$, \\
    $B_3 = \min_{a \in \{2,...,k\}}\{\gamma_{R3}^2 (H_a, v_a)+\sum_{r \in \{2, ..., k\}\setminus\{a\}}\min\{\gamma_{R3}^i(H_r, v_r):i=0,1,4,5,\\6,7,8\}\}$, \\
    $B_4 = \min_{a,b \in \{2,...,k\}}\{\min\{\gamma_{R3}^i (H_a, v_a): i=1,7,8\}+\{\min\{\gamma_{R3}^i (H_b, v_b): i=1,7,8\}+\sum_{r \in \{2, ..., k\}\setminus\{a,b\}}\min\{\gamma_{R3}^i(H_r, v_r):i=0,1,4,5,6,7,8\}\}$, \\
    $C_1 = \min_{a \in \{2,...,k\}}\{\gamma_{R3}^1 (H_a, v_a)+\sum_{r \in \{2, ..., k\}\setminus\{a\}}\min\{\gamma_{R3}^i(H_r, v_r):i=0,4\}\}$, \\
    $C_2 = \min_{a \in \{2,...,k\}}\{\min\{\gamma_{R3}^i (H_a, v_a): i= 1,7\}+\sum_{r \in \{2, ..., k\}\setminus\{a\}}\min\{\gamma_{R3}^i(H_r, v_r):i=0,1,4,5,7,8\}\}$, \\
    $C_3 = \min_{a \in \{2,...,k\}}\{\min\{\gamma_{R3}^i (H_a, v_a): i= 1,7,8\}+\sum_{r \in \{2, ..., k\}\setminus\{a\}}\min\{\gamma_{R3}^i(H_r, v_r):i=0,1,2,3,4,5,6,7,8\}\}$. \\
    % $D_1$ = $\min_{a\in \{2,...,k\}}\{\min\{\gamma_{R3}^i (H_a, v_a):i=2,3\}+\sum_{r \in \{2, ..., k\}\setminus\{a\}}\min\{\gamma_{R3}^i(H_r, v_r):i =0,1,2,3,4,5,6,7,8\}\}$, \\
    % $D = \min_{a\in \{2,...,k\}}\{\min\{\gamma_{R3}^i (H_a, v_a):i=1,2,3\}+\sum_{r \in \{2, ..., k\}\setminus\{a\}}\min\{\gamma_{R3}^i(H_r, v_r):i =0,1,2,3,4,5,6,7,8\}\}$.
\end{lemma}
\begin{proof}\textbf{(a)} We have the following four cases, based on the labels of the vertices of the set $\{v_2,v_3,...,v_k\}$. \vspace{2mm} \\
\textbf{Case 1.} $f(u) = 0$ for each $u \in \{v_2,v_3,...,v_k\}$ \vspace{2mm}\\
$f(v_1) = 0$ and $f(u) = 0$ for each vertex $u \in \{v_2,v_3,...,v_k\}$. Hence, we have that $\gamma_{R3}^0(H, v_1) = \sum_{1 \leq r \leq k} \gamma_{R3}^0 (H_r, v_r).$ \vspace{2mm} \\
\textbf{Case 2.} $f(u) = 1$ for exactly one vertex $u \in \{v_2,v_3,...,v_k\}$ and $f(v) = 0$ for each vertex $v \in \{v_2,v_3,...,v_k\}\setminus \{u\}$\vspace{2mm}\\
We have that $f(v) = 0$ for each $v \in \{v_2,v_3,...,v_k\}$ except for one vertex $u \in \{v_2,v_3,...,v_k\}$ such that $f(u) = 1$. As $f(v_1) = 0$, by Lemma \ref{lemma5a}(i), we have that $f_{H_1} \in \{D^0(H_1,v_1) \cup D^4(H_1, v_1)\}$. As one vertex $u \in \{v_2,v_3,...,v_k\}$ must be labeled 1, we will be looking for $\gamma_{R3}^1(H_a, v_a)$ for some $a \in \{2,3,...,k\}$. Hence, we obtain that $\gamma_{R3}^0(H, v_1) = \min \{\gamma_{R3}^i (H_1, v_1):i =0, 4 \}+C_1.$ \vspace{2mm} \\
\textbf{Case 3.} $f(u) = f(v) = 1$ for exactly two vertices $u,v \in \{v_2,v_3,...,v_k\}$ and $f(w) = 0$ for $w \in \{v_2,v_3,...,v_k\}\setminus \{u,v\}$
or $f(u) = 2$ for exactly one vertex $u \in \{v_2,v_3,...,v_k\}$ and $f(v) = 0$ for $v \in \{v_2,v_3,...,v_k\}\setminus \{u\}$\vspace{2mm}\\
We have that $f(w) = 0$ for each vertex $w \in \{v_2,v_3,...,v_k\}$ except for one vertex $u$ such that $f(u) = 2$ or except for two vertices $u, v$ such that $f(u) = 1$ and $f(v) = 1$. As $f(v_1) = 0$, by Lemma \ref{lemma5a}(ii) and Lemma \ref{lemma5a}(iii), we have that $f_{H_1} \in \{D^0(H_1,v_1) \cup D^4(H_1, v_1)\} \cup D^5(H_1,v_1)$. Hence, We obtain that $\gamma_{R3}^0(H, v_1) = \min \{\gamma_{R3}^i (H_1, v_1):i =0, 4, 5\}+\min\{B_1, B_2\}.$\vspace{2mm} \\
\textbf{Case 4.} $f(u) \geq 1, f(v) \geq 1, f(w) \geq 1$ for vertices $u,v,w \in \{v_2,v_3,...,v_k\}$
or $f(u) \geq 2$ for one vertex $u \in \{v_2,v_3,...,v_k\}$ and $f(v) \geq 1$ for one vertex $v \in \{v_2,v_3,...,v_k\}$ or $f(u) =3$ for one vertex $u \in \{v_2,v_3,...,v_k\}$\vspace{2mm} \\
We have that $f(x) = 0$ for each vertex $x \in \{v_2,v_3,...,v_k\}$ except for one vertex $u$ such that $f(u) = 3$ or except for two vertices $u, v$ such that $f(u) \geq 2$ and $f(v) \geq 1$ or except for three vertices $u, v, w$ such that $f(u) \geq 1$, $f(v) \geq 1$ and $f(w) \geq 1$. As $f(v_1) = 0$, by Lemma \ref{lemma5a}(iv), we have that $f_{H_1} \in \{D^0(H_1,v_1) \cup D^4(H_1, v_1)\} \cup D^5(H_1,v_1) \cup D^6(H_1,v_1)$. As either one vertex $u \in \{v_2,v_3,...,v_k\}$ must be labeled 3 or two vertices  $u,v \in \{v_2,v_3,...,v_k\}$ must be labeled 2 and 1 or three vertices  $u,v,w \in \{v_2,v_3,...,v_k\}$ must be labeled 1, we will be looking for $\gamma_{R3}^1(H_a, v_a)$ for three vertices $a \in \{2,3,...,k\}$ or $\gamma_{R3}^2(H_a, v_a)$ and $\gamma_{R3}^1(H_b, v_b)$ for vertices $a,b \in \{2,3,...,k\}$ or $\gamma_{R3}^3(H_a, v_a)$ for one vertex $a \in \{2,3,...,k\}$. 
Hence, we have that  $\gamma_{R3}^0(H, v_1) = \min \{\gamma_{R3}^i (H_1, v_1):i =0, 4, 5, 6\}+\min\{A_1,A_2,A_3\}.$ \vspace{4mm} \\
\textbf{(b)} We have the following three cases, based on the labels of the vertices of the set $\{v_2,v_3,...,v_k\}$. \\
\textbf{Case 1.} $f(u) = 0$ for each $u \in \{v_2,v_3,...,v_k\}$\vspace{2mm}\\
$f(u) = 0$ for each vertex $u \in \{v_2,v_3,...,v_k\}$, and $f(v_1) = 1$. Hence, we have that $\gamma_{R3}^1 (H, v_1)$ = $\gamma_{R3}^1 (H_1, v_1)+\sum_{r \in \{2, ..., k\}}\min\{\gamma_{R3}^i(H_r, v_r):i=0,4\}$. \vspace{2mm} \\
\textbf{Case 2.} $f(u) = 1$ for exactly one vertex $u \in \{v_2,v_3,...,v_k\}$ and $f(v) = 0$ for $v \in \{v_2,v_3,...,v_k\}\setminus \{u\}$\vspace{2mm}\\
We have that $f(v) = 0$ for each $v \in \{v_2,v_3,...,v_k\}$ except for one vertex $u \in \{v_2,v_3,...,v_k\}$ such that $f(u) = 1$. As $f(v_1) = 1$, by Lemma \ref{lemma5b}(i), we have that $f_{H_1} \in \{D^1(H_1,v_1) \cup D^7(H_1, v_1)\}$.  As one vertex $u \in \{v_2,v_3,...,v_k\}$ must be labeled 1, we will be looking for $\gamma_{R3}^1(H_a, v_a)$ for some $a \in \{2,3,...,k\}$. Hence, we obtain that $\gamma_{R3}^1 (H, v_1)$ is $\min \{\gamma_{R3}^i (H_1, v_1):i =1, 7\}+C_2$. \vspace{2mm} \\
\textbf{Case 3.} $f(u) \geq 1, f(v) \geq 1$ for two vertices $u,v \in \{v_2,v_3,...,v_k\}$
or $f(u) \geq 2$ for one vertex $u \in \{v_2,v_3,...,v_k\}$ \vspace{2mm}\\
We have that $f(w) = 0$ for each vertex $w \in \{v_2,v_3,...,v_k\}$ except for one vertex $u$ such that $f(u) \geq 2$ or except for two vertices $u, v$ such that $f(u) \geq 1$ and $f(v) \geq 1$. As $f(v_1) = 1$, by Lemma \ref{lemma5b}(ii), we have that $f_{H_1} \in \{D^1(H_1,v_1) \cup D^7(H_1, v_1)\} \cup D^8(H_1,v_1)$. As either one vertex $u \in \{v_2,v_3,...,v_k\}$ must be labeled 2 or two vertices  $u,v \in \{v_2,v_3,...,v_k\}$ must be labeled 1, we will be looking for $\gamma_{R3}^1(H_a, v_a)$ for two vertices $a \in \{2,3,...,k\}$ or $\gamma_{R3}^2(H_a, v_a)$ for one vertex $a \in \{2,3,...,k\}$. Hence, we have $\gamma_{R3}^1 (H, v_1)$ = $\min \{\gamma_{R3}^i (H_1, v_1):i =1, 7, 8\}+\min\{B_3,B_4\}$. \vspace{4mm} \\
\textbf{(c)} We have the following two cases, based on the labels of the vertices of the set $\{v_2,v_3,...,v_k\}$. \\
\textbf{Case 1.} $f(v) = 0$ for each $v \in \{v_2,v_3,...,v_k\}$\vspace{2mm}\\
As $f(v_1) = 2$, by Observation \ref{obs1} we have that, $\gamma_{R3}^2 (H, v_1) = \gamma_{R3}^2 (H_1, v_1)+\sum_{r \in \{2, ..., k\}}\min\{\gamma_{R3}^i(H_r, v_r):i=0,4,5\}$. \vspace{2mm} \\
\textbf{Case 2.} $f(v) \geq 1$ for a vertex $u \in \{v_2,v_3,...,v_k\}$ \vspace{2mm}\\
As $f(v_1) = 2$, by Observation \ref{obs1}, we have that $\gamma_{R3}^2 (H, v_1) = \gamma_{R3}^2 (H_1, v_1)$. As one vertex $u \in \{v_2,v_3,...,v_k\}$ must be labeled 1, we will be looking for $\gamma_{R3}^1(H_a, v_a)$ for some $a \in \{2,3,...,k\}$. Hence, we obtain that $\gamma_{R3}^2 (H, v_1) = \gamma_{R3}^2 (H_1, v_1)+C_3$. \vspace{4mm} \\
\textbf{(d)} As $f(v_1) = 3$, by Observation \ref{obs1} we have that, $\gamma_{R3}^3(H, v_1)$ = $\gamma_{R3}^3 (H_1, v_1)+\sum_{r \in \{2, ..., k\}}\min\{\gamma_{R3}^i(H_r, v_r):i=0,1,2,3,4,5,6,7,8\}.$ \vspace{4mm}\\
\textbf{(e)} If there exists either a vertex $u \in \{v_2,v_3,...,v_k\}$ with $f(u) = 2$ or a pair of vertices $u,v \in \{v_2,v_3,...,v_k\}$ such that $f(u) \geq 1$ and $f(v) \geq 1$, then we have that $\gamma_{R3}^4(H, v_1) = \gamma_{R3}^4 (H_1, v_1)+\sum_{r \in \{2, ..., k\}}\gamma_{R3}^0(H_r, v_r)$. \vspace{2mm}\\
If there exists a vertex $u \in H_1$ with $f(u) = 1$, then we have that $\gamma_{R3}^4(H, v_1) = \gamma_{R3}^5 (H_1, v_1)+C_1$. \vspace{2mm}\\
If each vertex from $H_1$ has a label of zero, then we have that $\gamma_{R3}^4(H, v_1) = \gamma_{R3}^6 (H_1, v_1)+\min\{B_1,B_2\}$. \vspace{4mm} \\
\textbf{(f)} If there exists a vertex from $u \in H_1$ with $f(u) = 1$, we have that $\gamma_{R3}^5 (H, v_1) = \gamma_{R3}^5 (H_1, v_1)+\sum_{r \in \{2, ..., k\}}\gamma_{R3}^0(H_r, v_r)$. \vspace{2mm}\\
        If each vertex from $H_1$ has a label of zero, then we have that $\gamma_{R3}^5 (H, v_1) = \gamma_{R3}^6 (H_1, v_1)+C_1$. \vspace{4mm} \\
\textbf{(g)} We have that $\gamma_{R3}^6(H, v_1)$ = $\gamma_{R3}^6 (H_1, v_1)+\sum_{r \in \{2, ..., k\}}\gamma_{R3}^0(H_r, v_r)$. \vspace{4mm} \\
\textbf{(h)} If there exists a vertex from $u \in H_1\setminus\{v_1\}$ with $f(u) = 1$, then we have that $\gamma_{R3}^7 (H, v_1) = \gamma_{R3}^7 (H_1, v_1)+\sum_{r \in \{2, ..., k\}}\gamma_{R3}^0(H_r, v_r)$. \vspace{2mm}\\
         If each vertex from $H_1$ has a label of zero, then we have that $\gamma_{R3}^7 (H, v_1) = \gamma_{R3}^8 (H_1, v_1)+C_2$. \vspace{4mm} \\
\textbf{(i)} We have that $\gamma_{R3}^8(H, v_1)$ = $\gamma_{R3}^8 (H_1, v_1)+\sum_{r \in \{2, ..., k\}}\gamma_{R3}^0(H_r, v_r)$. \qed\end{proof}
We present the algorithm \hyperref[alg1]{\texttt{R3DN-BLOCK}(\(G\))} to compute \( \gamma_{R3}(G) \) for a given connected block graph \( G \), using a dynamic programming approach.
The algorithm maintains a working graph \( G' \), which is initially set to \( G \). The values of \( \gamma_0, \gamma_1, \gamma_2,\) \(\gamma_3, \gamma_4, \gamma_5 \), \(\gamma_6\), \(\gamma_7\) and \(\gamma_8\) are initialized at each vertex \( v \in V(G) \), corresponding to the subgraph \( G[\{v\}] \).
At each iteration, the algorithm selects an end block of \( G' \). Suppose at a particular iteration, the block \( \mathcal{B} \) is selected, where \( V(\mathcal{B}) = \{v_1, v_2, \ldots, v_k\} \) and \( v_1 \) is the cut-vertex connecting \( \mathcal{B} \) to the rest of \( G' \). The algorithm computes the values of \( \gamma_0, \gamma_1, \gamma_2, \gamma_3, \gamma_4, \gamma_5 \), \(\gamma_6\), \(\gamma_7\) and \(\gamma_8\) for this block using Lemma \ref{lemma6}, and stores them at the cut-vertex \( v_1 \). The block \( \mathcal{B} \) is then removed from \( G' \).
This process is repeated until only one block remains in \( G' \), say \( \mathcal{B}_q \). At this point, the algorithm selects an arbitrary vertex \( v_r \in V(\mathcal{B}_q) \) to store the final computed values. Finally, using Lemma \ref{lemma4}, we determine \( \gamma_{R3}(G) \) from the values stored at \( v_r \).
\begin{algorithm} [t] \label{alg1}
    \caption{\texttt{R3DN-BLOCK}(\(G\))}
    {Input: block graph $G$\;
    Output: $\gamma_{R3}{G}$\;}
    $G' = G$\;
    $S = \emptyset$\;
    \ForEach {$v \in V(G)$} {
    $\gamma_{R3}^0(G[\{v\},v]) = \infty$\;
    $\gamma_{R3}^1(G[\{v\},v]) = \infty$\;
    $\gamma_{R3}^2(G[\{v\},v]) = 2$\;
    $\gamma_{R3}^3(G[\{v\},v]) = 3$\;
    $\gamma_{R3}^4(G[\{v\},v]) = \infty$\;
    $\gamma_{R3}^5(G[\{v\},v]) = \infty$\;
    $\gamma_{R3}^6(G[\{v\},v]) = 0$\;
    $\gamma_{R3}^7(G[\{v\},v]) = \infty$\;
    $\gamma_{R3}^8(G[\{v\},v]) = 1$\;
    }
    \While {$G'$ has more than one vertex} {
        Choose an end block $\mathcal{B}$ in $G'$, where $V(\mathcal{B}) = \{v_1, v_2, ..., v_k\}$ and $v_1$ is the only cut-vertex when $\mathcal{B}$ has a cut-vertex\;
        $S = S \cup V(\mathcal{B})$\;
        Let $H_i$ be the component in $G[S] - E(\mathcal{B})$ such that $V(\mathcal{B}) \cap V(H_i) = \{v_i\}$ for each $i \in [k]$\;
        Let $H$ be the graph constructed from $H_1, H_2, ..., H_k$ by adding edges to make $\{v_1, v_2, ..., v_k\}$ a clique in $H$\;
        Determine $\gamma_{R3}^i(H,v_1)$ for all $i \in \{0\} \cup [8]$ using Lemma \ref{lemma6}\;
        $G' = G'-\{v_2, v_3, ..., v_k\}$
        }
    return $\min\{\gamma_{R3}^0(G, v_r)$, $\gamma_{R3}^1 (G, v_r)$, $\gamma_{R3}^2 (G, v_r)$, $\gamma_{R3}^3 (G, v_r)\}$, where $v_r$ is the only vertex in $G'$.
    % return $\gamma_{R3}(G)$. 
\end{algorithm} 
\begin{lemma}\label{blockgraph}
    Algorithm \hyperref[alg1]{\texttt{\textup{R3DN-BLOCK}}\textup{(\(G\))}} computes $\gamma_{R3}(G)$ of a block graph in $\mathcal{O}(n^3)$ time.
\end{lemma}
\begin{proof}Let \( G \) be a connected block graph. Let the set of cut-vertices of \( G \) be denoted by \( \{c_1, c_2, \ldots, c_s\} \), and the set of blocks be denoted by \( \{\mathcal{B}_1, \mathcal{B}_2, \ldots, \mathcal{B}_q\} \). We initialize the values of \( \gamma_0, \gamma_1, \gamma_2, \gamma_3, \gamma_4, \gamma_5, \gamma_6, \gamma_7 \) and \(\gamma_8\) for every vertex in \( G \) in \( \mathcal{O}(n) \) time.
At each iteration, the algorithm selects a end block in the current graph \( G' \) based on the ordering obtained by performing a breadth-first search (BFS) on the \textit{cut-tree} \( T_G \).
Constructing \( T_G \) takes \( \mathcal{O}(|V(G)| + |E(G)|) \) time~\cite{aho1974design}, and a BFS traversal of \( T_G \) takes \( \mathcal{O}(q + s) \) time. Hence, the ordering of vertices can be computed in overall \( \mathcal{O}(|V(G)| + |E(G)|) \) time.
We now analyze the time complexity of a single iteration of the \texttt{while} loop. Let \( \mathcal{B}_i \) be the block selected during the \( i \)th iteration, with vertex set \( V(\mathcal{B}_i) = \{v_1, v_2, \ldots, v_k\} \). \vspace{2mm} \\
Now we calculate the time required for computing \( \gamma_0, \gamma_1, \gamma_2, \gamma_3, \gamma_4, \gamma_5, \gamma_6, \gamma_7\) and \(\gamma_8\). \vspace{2mm} \\
\textbf{Computing $A_1$ (or $B_1$, $B_3$, $C_1$, $C_2$, $C_3$):} \vspace{2mm} \\
We have that $A_1$ = $\min_{a \in \{2,...,k\}}\{\gamma_{R3}^3 (H_a, v_a)+\sum_{r \in \{2, ..., k\}\setminus\{a\}}\min\{\gamma_{R3}^i(H_r, v_r):i=0,1,2,3,4,5,6,7,8\}\}$. We first compute the value of $Y$ = $\sum_{r \in \{2, ..., k\}}\min$ $\{\gamma_{R3}^i(H_r, v_r):i=0,1,2,3,4,5,6,7,8\}$ in $\mathcal{O}(|V(\mathcal{B}_i)|)$ time. Next, we compute the value of $\gamma_{R3}^3 (H_a, v_a)+(Y-\min\{\gamma_{R3}^i(H_a, v_a):i=0,1,2,3,4,5,6,7,8\})$ for each vertex $v \in \{v_2,v_3,...,v_k\}$ and chose the minimum among them to obtain $A_1$. Hence, we compute the value of $A_1$ in $\mathcal{O}(|V(B_i)|)$ time. As $B_1$, $B_3$, $C_1$, $C_2$ and $C_3$ have similar form to that of $A_1$, the values of $B_1$, $B_3$, $C_1$, $C_2$ and $C_3$ can also be computed in $\mathcal{O}(|V(\mathcal{B}_i)|)$ time. \vspace{2mm} \\
\textbf{Computing $A_2$ (or $B_2$, $B_4$):} \vspace{2mm} \\
Recall that, $A_2$ = $\min_{a,b \in \{2,...,k\}}\{\gamma_{R3}^2 (H_a, v_a)+\min\{\gamma_{R3}^i (H_b, v_b): i=1,7,8\}+\sum_{r \in \{2, ..., k\}\setminus\{a,b\}}\min$ $\{\gamma_{R3}^i(H_r, v_r):i=0,1,2,3,4,5,6,7,8\}\}$. We compute the value of $Y = \sum_{r \in \{2, ..., k\}}$ $\min\{\gamma_{R3}^i(H_r, v_r):i=0,1,2,3,4,5,6,7,8\}$ in $\mathcal{O}(|V(\mathcal{B}_i)|)$ time. Next, we compute the value of $\gamma_{R3}^2 (H_a, v_a)+\min\{\gamma_{R3}^i (H_b, v_b): i=1,7,8\}+(Y-\sum_{r \in \{a,b\}}$ $\{\gamma_{R3}^i(H_r, v_r):i=0,1,2,3,4,5,6,7,8\})$ for each pair of vertices $v_a,v_b \in \{v_2,v_3,...,v_k\}$ and chose the minimum among them to obtain $A_2$. The number of such pairs is $|E(\mathcal{B}_i)|$. Hence, we compute the value of $A_2$ in $\mathcal{O}(|V(\mathcal{B}_i)|+|E(\mathcal{B}_i)|)$ time. As $B_2$ and $B_4$ have similar form to that of $A_2$, the values of $B_2$ and $B_4$ can be computed in $\mathcal{O}(|V(\mathcal{B}_i)|+|E(\mathcal{B}_i)|)$ time. \vspace{2mm} \\
\textbf{Computing $A_3$:} \vspace{2mm} \\
We have that $A_3$ = $\min_{a,b,c \in \{2,...,k\}}\{\min\{\gamma_{R3}^i (H_a, v_a): i=1,7,8\}+\{\min\{\gamma_{R3}^i $ $ (H_b, v_b): i=1,7,8\}+\{\min\{\gamma_{R3}^i (H_c, v_c): i=1,7,8\}+\sum_{r \in \{2, ..., k\}\setminus\{a,b,c\}}\min\{\gamma_{R3}^i $ $(H_r, v_r):i=0,1,2,3,4,5,6,7,8\}\}$. We compute the value of $Y = \sum_{r \in \{2, ..., k\}}$ $\min \{\gamma_{R3}^i(H_r, v_r):i=0,1,2,3,4,5,6,7,8\}$ in $\mathcal{O}(|V(\mathcal{B}_i)|)$ time. Next, we compute the value of $\{\gamma_{R3}^i (H_a, v_a): i=1,7,8\}$+$\{\min\{\gamma_{R3}^i (H_b, v_b): i=1,7,8\}+\{\min\{\gamma_{R3}^i (H_c, v_c): i=1,7,8\}$+$(Y-\sum_{r \in \{a,b,c\}}\{\gamma_{R3}^i(H_r, v_r):i=0,1,2,3,4,5,6,7,8\})$ for each triple of vertices $v_a,v_b,v_c \in \{v_2,v_3,...,v_k\}$ and chose the minimum among them to obtain $A_3$. The number of such triples is $|V(\mathcal{B}_i)| \choose 3$ and the value of $A_3$ can be computed in $\mathcal{O}(|V(\mathcal{B}_i)|^3)$ time. \vspace{2mm} \\
% \textbf{Computing $D$ (or $E$):} \vspace{2mm} \\
% Recall that,
%     $D$ = $\min_{a\in \{2,...,k\}}\{\min\{\gamma_{R3}^i (H_a, v_a):i=2,3\}+\sum_{r \in \{2, ..., k\}\setminus\{a\}}\min\{\gamma_{R3}^i(H_r, v_r):i =0,1,2,3,4,5,6,7,8\}\}$. We compute the value of $Y = \sum_{r \in \{2, ..., k\}}\min\{\gamma_{R3}^i(H_r, v_r):i=0,1,2,3,4,5,6,7,8\}$ in $\mathcal{O}(|V(B_i)|)$ time. Next, we compute $\min\{\gamma_{R3}^i (H_a, v_a):i=2,3\}+(Y-\min\{\gamma_{R3}^i(H_a, v_a):i =0,1,2,3,4,5,6,7,8\})$ for each vertex $v \in \{v_2,v_3,...,v_k\}$ and chose the minimum among them to obtain $D$. Hence, we compute the value of $A_1$ in $\mathcal{O}(|V(B_i)|)$ time. As $E$ has similar form to that of $D$, the value of $E$ can be computed in $\mathcal{O}(|V(B_i)|+|E(B_i)|)$ time. \vspace{2mm} \\
Hence, we obtain that Algorithm {\hyperref[alg1]{\texttt{R3DN-BLOCK}(\(G\))}} computes $\gamma_{R3}(G)$ of a block graph in $\mathcal{O}(n^3)$ time. \qed\end{proof}
Hence, we obtain the following theorem.
\begin{theorem}
    R3D on block graphs is polynomial-time solvable.
\end{theorem}
% \section{Complexity difference between \rd{} and \rtd{}}
% \input{6_complexityDiff}
\section{Conclusion}
In this paper, we have investigated the complexity of R3D and proved that the problem is NP-complete on split graphs, a subclass of chordal graphs. We have initiated the study on the parameterized complexity of the problem and showed the problem is W[2]-hard parameterized by weight. In addition, we presented a polynomial-time algorithm for block graphs. We conclude the paper with the following open questions.
\begin{itemize}
    \item What is the complexity of the problem on interval graphs, strongly chordal graphs and AT-free graphs?
    \item What is the parameterized complexity of the problem parameterized by weight on bipartite graphs, split graphs, planar graphs and circle graphs?
    \item Does R3D admit single-exponential time parameterized algorithms for vertex cover number, distance to clique, neighbourhood diversity and distance to cluster?
\end{itemize}
\bibliographystyle{splncs04}
\bibliography{reference}
\end{document}